\documentclass[10pt,journal,compsoc]{IEEEtran}

\usepackage{amsmath}
\usepackage{amsthm}
\usepackage{algorithm} 
\usepackage{algpseudocode}
\usepackage{lipsum}% http://ctan.org/pkg/lipsum
\usepackage{multicol}% http://ctan.org/pkg/multicols

\usepackage{caption}
\usepackage{subcaption}
\usepackage{color}
\usepackage{multirow,rotating} % Table I 
\usepackage{subfiles} \usepackage{hyperref}
\newtheorem{theorem}{Theorem}

\newcommand{\BO}[1]{{ O}\left(#1\right)}

\newcommand{\BT}[1]{{\Theta}\left(#1\right)}

\begin{document}
\title{Blocking Techniques for Sparse Matrix Multiplication on Tensor Accelerators}

\author{Paolo~Sylos~Labini, Massimo~Bernaschi, Francesco~Silvestri, Flavio~Vella% <-this % stops a space
\IEEEcompsocitemizethanks{\IEEEcompsocthanksitem 
Paolo Sylos Labini is with Free University of Bozen-Bolzano, Italy. \protect\\
E-mail: paolo.syloslabini@stud-inf.unibz.it
\IEEEcompsocthanksitem Francesco Silvestri is with the University of Padova, Italy. \protect% <-this % stops an unwanted space
\IEEEcompsocthanksitem Massimo Bernaschi is with the IAC-CNR, Italy
\IEEEcompsocthanksitem Flavio Vella is with the University of Trento, Italy. \protect
% <-this % stops an unwanted space
}
}

\markboth{Journal of \LaTeX\ Class Files,~Vol.~14, No.~8, August~2015}%
{Shell \MakeLowercase{\textit{et al.}}: Bare Demo of IEEEtran.cls for Computer Society Journals}



\IEEEtitleabstractindextext{%
\begin{abstract}
Tensor accelerators have gained popularity because they provide a cheap and efficient solution for speeding up computational-expensive tasks in Deep Learning and, more recently, in other Scientific Computing applications. 
However, since their features are specifically designed for tensor algebra (typically dense matrix-product), it is commonly assumed that they are not suitable for applications with sparse data.
To challenge this viewpoint, we discuss methods and present solutions for accelerating sparse matrix multiplication on such architectures.
In particular, we present a 1-dimensional blocking algorithm with theoretical guarantees on the density, which builds dense blocks from arbitrary sparse matrices. Experimental results show that, even for unstructured and highly-sparse matrices, our block-based solution which exploits Nvidia Tensor Cores is faster than its sparse counterpart.
We observed significant speed-ups of up to two orders of magnitude on real-world sparse matrices. 
\end{abstract}

\begin{IEEEkeywords}
Tensor Architectures; Sparse Matrix Operations; Algorithms for block reordering; GPU Computing;
\end{IEEEkeywords}}

\maketitle

\IEEEdisplaynontitleabstractindextext

% \fi
\IEEEpeerreviewmaketitle

\IEEEraisesectionheading{\section{Introduction}\label{sec:introduction}}


\IEEEPARstart{W}{ith} the advent of Deep Learning (DL) as the mainstream methodology for extracting new knowledge from massive amounts of data, a plethora of special-purpose architectures have been developed for accelerating the training and inference phases of Deep Neural Networks~(DNN). Most popular examples include accelerators like the Tensor Processing Unit~(TPU)~\cite{jouppi2017datacenter} or the Intelligent Processing Unit~(IPU)~\cite{graphcore-ipu} and compute units like the Nvidia Tensor Cores~(TCs)~ \cite{nvidia-ampere} among others~\cite{9286149}. \textit{Tensor accelerators} target the most expensive operation in DNNs, namely the multiply-and-accumulate operation, and specialize in the parallel multiplication of large batches of small dense matrices. Recently, their design and capabilities also enable the acceleration of several fundamental primitives of traditional scientific applications~\cite{10.1145/3437801.3441599, ROMERO2020107473,dakkak2019accelerating}, offering challenges and opportunities from the theoretical viewpoint~\cite{CSV20, CSV21}.

However, a broad class of applications, including DNNs themselves, Graph Neural Networks and Graph Analytics, is characterized by sparse input domains, resulting in highly-irregular computations and memory access patterns.
This limits the scalability of parallel algorithms with low arithmetic intensity, such as sparse matrix multiplication (SpMM), even on architectures that offer high-bandwidth memories like Graphics Processing Units (GPU)~\cite{10.1145/1504176.1504181}.
From the architecture design perspective, there are examples of highly-specialized hardware that natively perform sparse multiplication~\cite{9065428,scheff-stream}, but most of them come with several limitations regarding the data-layout (e.g. they assume squared matrices only) or the sparsity pattern (e.g. structured sparsity only~\cite{nvidia-ampere}). This limits their adoption and diffusion~\cite{survGNN-acc-surv}. 
For all these reasons, high-performance SpMM solutions tend to be extremely customized, over-fitted to the problem, data and hardware at hand~\cite{FilipponeReview,10.1145/1654059.1654078}. This results, inevitably, in a lack of performance portability and usability.
The diffusion and cost-effectiveness of tensor accelerators, as well as the availability of libraries for dense matrix multiplication, suggest investigating how to exploit such architectures to accelerate SpMM.
Intuitively, solutions to this problem would cluster the nonzero elements of the matrix and feed them to compute units specialized for the dense product. However, there are many possible ways to implement such a "compression". One approach would consist of explicitly storing the coordinates of non-zero elements of the matrix into external \textit{ad-hoc} data-structures, performing the dense multiplication and then storing back the results~\cite{accTensor2020}. 
An alternative approach consists in looking for dense, or almost dense substructures of the input matrix. This allows re-using higher-level routines that are already optimized for dense multiplication, on top of the tensor accelerator. Unfortunately, dense substructures in sparse matrices are often hidden by the arbitrary labelling of rows and columns. Here, blocking algorithms try to solve this by permuting and grouping along the dimensions of the matrix.

Most existing blocking algorithms focus on symmetric matrices and symmetric reorderings, treating the matrix as the adjacency matrix of a graph to reduce the number of iterations of numerical solvers~\cite{CMReorder,PABLO, Saad}. A straightforward application of such methods work poorly on matrices with arbitrary shapes (e.g. sparse layers in DNNs) and sparsity patterns (e.g. real-world graphs). Furthermore, they only partially control the size and internal sparsity of the blocks, failing to consistently improve the performance of SpMM on parallel architectures~\cite{PichelComparison}.

To address these challenges, in this paper we present a method to build blocks of tunable size with theoretical guarantees on the density. 
Specifically, we designed a 1-dimensional reordering algorithm that overcomes the aforementioned limitations. We discuss under what conditions (block dimensions and density) this approach can efficiently exploit tensor accelerators for spMM.
Finally, we present evidence that, for a wide class of sparse matrices, the recent advancements in dense accelerators make it convenient to use traditional dense-specific routines instead of sparse-specific ones.  
Our contributions can be summarized as follows:

\begin{itemize}
    \item We present 1-SA, a 1-dimensional blocking algorithm that can decompose a general sparse matrix in dense blocks, with control over the size and density of blocks.
    \item We provide theoretical guarantees on the density of the reordered matrix produced by the 1-dimensional blocking algorithm.
    We also asymptotically analyze the running time of a sparse-dense matrix multiplication when using a reordered sparse matrix, by using the TCU computational model for Tensor Cores~\cite{CSV21,AS20}.
    \item We analyze the efficacy of 1-SA on a large synthetic dataset and we provide experimental evidence that even for very sparse matrices, using a block-based SpMM routine on tensor accelerators is more efficient than using a sparse-specific one in the Nvidia hardware/software stack. Specifically, we observe significant performance improvement (from 3x up to 100x) on real-world sparse matrices when using cuBLAS compared to cuSPARSE.

\end{itemize}

The manuscript is organized as follows. In Section~\ref{sec: background}, we provide the background and notation used for presenting the methodology and the algorithms discussed in Section~\ref{sec: methods-and-algos}. Experimental evaluation of both the blocking algorithm and its effect on multiplication performance is addressed in Section~\ref{sec: exp}. Finally, we provide an overview of other relevant works in Section~\ref{sec:related} and we draw conclusions in Section~\ref{sec:conclusion}.

\section{Background}
\label{sec: background}
The development of new sparse multiplication kernels is tightly related to sparse storage formats used for representing the data. 

Dense storage formats allow for good spatial and temporal locality when multiplying, and are especially effective for high-performance parallel multiplication, but have an obvious downside --- the explicit storing (and processing) of all the zero entries. For very sparse matrices, these can be several orders of magnitude more than the nonzeros. In most practical situations, the additional storage space and multiplication workload eat away any benefit
coming from the data regularity.
Efficient sparse storage formats provide ways to store only, or mostly, nonzero
entries. The most popular, general purpose storage format is the Compressed Sparse Rows
(CSR) format. CSR stores only the nonzero entries, row by row, in a contiguous array, and uses two additional arrays to keep track of their positions.
However, the flexible structure of CSR makes it hard to optimize spMM on the GPU, due to low spatial and temporal locality and unbalanced workloads \cite{designSPMM}. 
A great deal of research has been focused on mitigating these issues, either by mapping a CSR multiplication efficiently on the GPU~\cite{designSPMM, spMMHongTiling} or by devising modified versions of the CSR format, such as CSR5~\cite{liu2015csr5}. 

Blocked storage formats achieve a middle ground between sparse and dense storage formats. These formats, such as the blocked CSR (BCSR)~\cite{VuducPHD}, are a popular choice for spMM on the GPU, because they reduce the indexing overhead and improve data reuse for suitable matrices~\cite{FilipponeReview}, but they pose unique challenges. \textit{Fill-in} (zeros stored as nonzeros) is an inevitable consequence of most types of blocking which increases their memory and computational footprint, particularly on unstructured sparse matrices. Flexible storage formats such as Unaligned BCSR (UBCSR)~\cite{VuducPHD} may reduce fill-in, but come with the costs of a heavier indexing structures. 
In this work, we will employ the Variable Block Row (VBR)~\cite{Saad94sparskit:a} format, a variant of BCSR which allows for blocks of variable size. To further reduce fill-in, we will make use of reordering and blocking techniques that rearrange the nonzero entries to expose dense substructures within the matrix.

\subsection{Similarity-based blocking}
Our analysis starts from an algorithm developed by Saad~\cite{Saad} for the preconditioning of iterative solvers, in turn based on the seminal work of Ashcraft~\cite{AshMinimumDegree}. Ashcraft's algorithm, briefly exemplified in Algorithm~\ref{algo:hash}, uses hashing to identify (and group together) identical rows (and columns) in a symmetric matrix. Then, by reordering the rows and columns according to the resulting partition, it clusters the nonzeros in perfectly dense blocks. Saad's version further compresses the symmetric matrix by also grouping rows and columns which are not identical, but just similar. Thus, it creates blocks which are only approximately dense. Unfortunately this approach, while effective in creating big, almost dense blocks, is not directly applicable to rectangular or asymmetric matrices. In this paper, we will consider instead 1-dimensional blocking algorithms, that is, those that only reorder one dimension (e.g., the rows) of the matrix leaving the other dimension untouched. The benefit of 1-dimensional blocking is that, when multiplying our sparse matrix by a dense one, we can leave the original order of the latter intact. This saves us the cost of reordering the dense matrix at run-time. In general, 1-dimensional reorderings allows us to keep any predetermined convenient ordering and grouping of the columns intact. In Section~\ref{sec: methods-and-algos} we will detail 1-SA, a modification of Saad's algorithm that reorders one dimension of a rectangular sparse matrix, and forces blocks of fixed size on the other dimension.
We start by discussing the algorithm from Saad~\cite{Saad}. We will then show how to extend Saad's algorithm to the case of 1-dimensional blocking. 

Saad's algorithm (SA) takes as input a symmetric matrix $A$ and a similarity parameter $\tau$. The first step in SA is a graph compression technique due to Ashcraft~\cite{AshMinimumDegree}. This technique amounts to hash each row so that identical rows are binned together. Ashcraft propose the simple but effective hash function
\begin{equation}
    h(v_i) = \sum_{ v_{i,j} \neq 0} {j}
    \label{hash}
\end{equation}
where $v$ is a row and $i$ and $j$ are respectively row and column indices. 
Eq. \eqref{hash} will hash identical rows to the same value. After checking for hash collisions, partitioning identical rows (and columns) of $A$ together will expose perfectly dense blocks in $A$. The hash function in \eqref{hash} is quite prone to collisions, but it has the advantage of being very fast to compute. As a matter of fact, the time spent checking for collisions is negligible compared to other steps of the procedure. Moreover, collisions can be reduced by taking in account the size of the rows in addition to their hash.

\begin{algorithm}
	\caption{Hash-based compression} 
	Input: A CSR matrix.
	\\Output: a mapping from rows to groups
	\begin{algorithmic}[1]
	\State Compute $\mathrm{h}(v_i)$ for each row $v_i$
	\State Sort rows by $\mathrm{h}(v_i)$
	\State Set $\mathrm{group}(v_i) = -1 \; \mathrm{for} \; i = 1,\ldots,N$
	\For {$i=1,\ldots, N$}
	    \If{$\mathrm{group}(v_i) \neq -1$} \textbf{skip}
	    \EndIf
	    \State $\mathrm{group}(v_i) = i$
        \For {$j = i+1, \ldots N$}
            \If{$h(v_j) \neq h(v_i)$} \textbf{break}
        	\Else 
        	    \If {$\mathrm{group}(v_j) = -1$}
        	        \If {$\mathrm{pattern}(v_i)$ = $\mathrm{pattern}(v_j)$} 
        	            \State $\mathrm{group}(v_j) = i$;
        	        \EndIf
        	    \EndIf
        	\EndIf
        \EndFor
	\EndFor
	\end{algorithmic} 
	\label{algo:hash}
\end{algorithm}

In some matrices an implicit block structure exists but is not perfect, meaning that many nodes are similar but not identical. Uncovering these imperfect substructures may allow to find larger blocks with a moderate loss in density. SA measures the similarity of nodes with cosine similarity, namely
\begin{equation}
    cos(v,w) = \frac{\langle v, w \rangle}{\|v\| \|w\|}.
    \label{cosinesim}
\end{equation}
Other similarity measures may be easily substituted in place of Eq. \ref{cosinesim}. In our algorithm, we will use Jaccard similarity instead, which provides us with better theoretical bounds. The jaccard similarity of two sets A, B (in this case, the sets of nonzero indices in two rows) is defined as: 
\begin{equation}
    jaccard(A,B) = \frac{|A\cap B|}{|A\cup B|}.
    \label{jaccard_sim}
\end{equation}

After compressing the graph with Ashcraft's method, SA compares the (compressed) nodes against each other, and merges them if their cosine similarity exceeds $\tau$. SA groups rows and columns in the same way to generate a blocking of the symmetric matrix.
 
SA can be readily adapted to rectangular matrices, as briefly mentioned in Saad's paper~\cite{Saad}. It suffices to drop the assumption of symmetry, and just compress and compare the rows (or the columns) of the matrix using Eq.~\eqref{hash} and Eq.~\eqref{cosinesim}. Ater the rows have been blocked, the columns may be then partitioned (e.g., uniformly) to create rectangular blocks. Unfortunately, this naive 1-dimensional implementation of SA is not nearly as effective as its symmetric counterpart, as it fails to group together elements which are close but not exactly on the same column. This limitation is addressed by our algorithm, 1-SA, which first partition the columns and then reorders the rows, thus accounting for entries with close indices when comparing rows.
\begin{figure}
    \centering
    \includegraphics[width = 0.8 \linewidth]{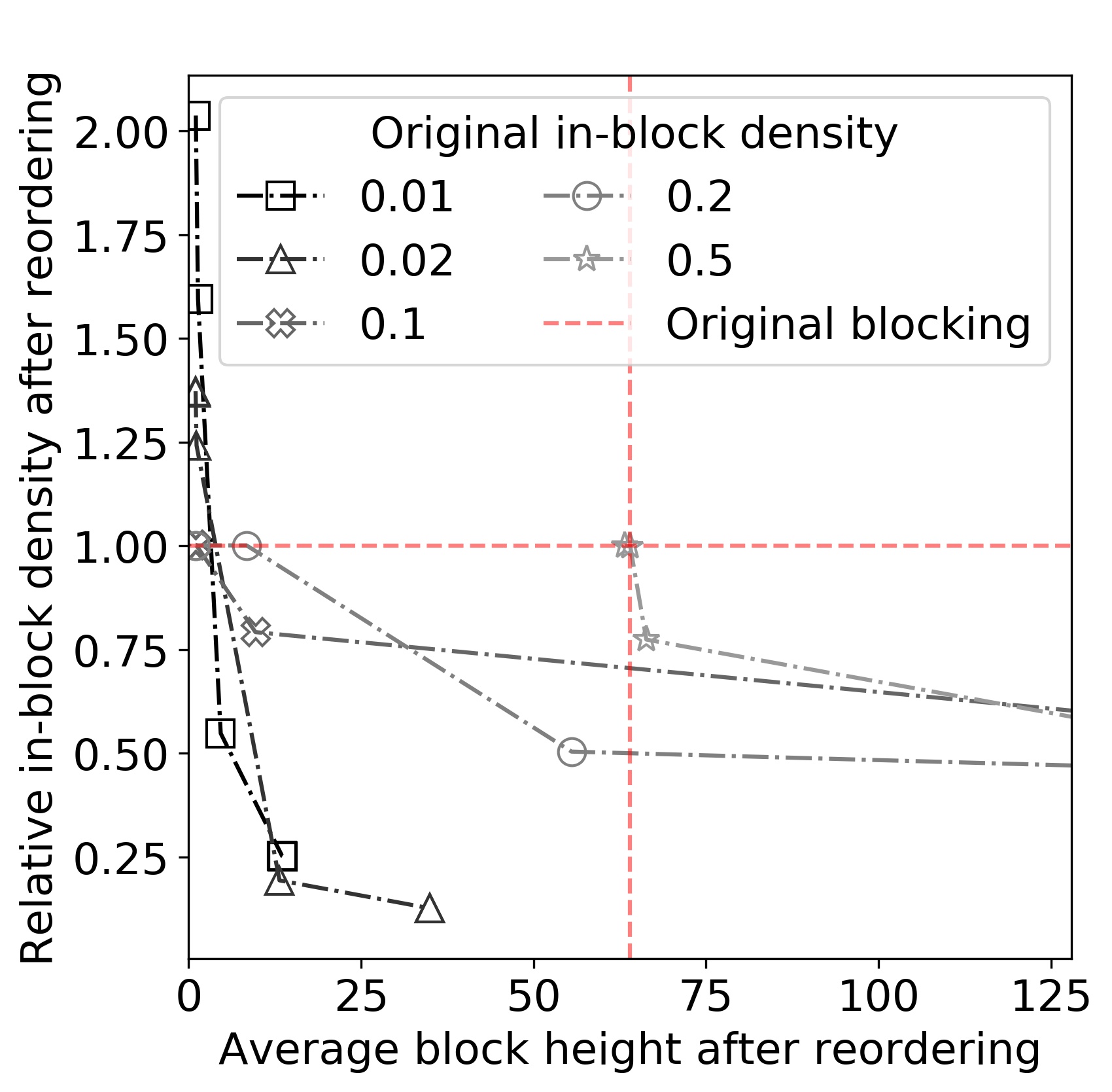}
    \caption{Blocking curves for SA on five synthetic blocked matrices, showing the blocking trade-off between block size (x-axis) and in-block density (y-axis). The $8192 \times 8192$ matrices are divided in blocks of size $64 \times 64$, with $10\%$ randomly chosen nonzero blocks. Each matrix has a different in-block density. Notice that SA can recover the original blocking (red cross) only for a very dense matrix.}
    \label{fig:saad_motivation}
\end{figure}

\subsection{Evaluating the quality of blockings}
\label{sec:blocking-quality}
To evaluate the effectiveness of a blocking algorithm, two aspects must be considered: the size of the blocks created, and their density. Coarser blockings reduce indexing and promote regular data access, but also introduce more fill-in, i.e. treating zero elements as non-zeros. On the other hand, finer blockings are usually denser at the cost of less data locality. In general, a trade-off exists between the size and density of the blocks. Ultimately the target task --- in our case, spMM --- will determine which balance between the two properties is the most desirable. In general, a good blocking algorithm should be able to explore the density-size trade-off to attain a satisfying density for a target block size.
Both SA and 1-SA are equipped with a parameter $\tau$ to explore this trade-off. Varying $\tau$ allows us to tune the density of the blocks, and indirectly, their size.
For this reason, we will not evaluate our algorithm on single-shot blockings, but on entire \textit{blocking curves} --- that is, we will observe how different choices of $\tau$ produce blockings with different combinations of block size and in-block density.
Since we will explore 1-dimensional reordering algorithms, only the height of blocks will be allowed to change, while the width remains fixed by a predetermined partition of the columns. 
Figure~\ref{fig:saad_motivation} shows an example of a reordering curve for a direct 1-dimensional implementation of SA, where several synthetic blocked matrices have been scrambled and then blocked again. For all but the denser matrices, SA failed to recover the original blocking. As we will see in section \ref{sec: exp:reorder}, these results can be improved by tweaking SA to account for the predetermined column partition. 

\begin{table}[htb!]
\centering
\footnotesize
\sf
\begin{tabular}{@{}l|ll@{}}
\hline
\\
\multirow{11}{*}{\begin{turn}{90}\shortstack{Blocked  matrices}\end{turn}} 
& $v_i$ & \textit{The i-th row of a sparse matrix}.\\
& $v_{i,j}$ & \textit{The j-th element of row $v_i$ (usually 0 or 1)}.\\
& $\hat{v}$ & \textit{A row projected on the column partition}.\\
& $k_i$ & \textit{Number of nonzero in row i}. \\
& $\Delta_W$ & \textit{Width of the column partition}\\
& $\Delta_H$ & \textit{Height of the blocks}.\\
& $\delta$  & \textit{Overall matrix density}.\\
& $\theta$ & \textit{Fraction of nonzero blocks}.\\
& $\rho$ & \textit{The overall density inside nonzero blocks} \\
& $\rho',\Delta_H'$ & \textit{As above, but after blocking} \\
\\
\hline
\\
\multirow{4}{*}{\begin{turn}{90}\shortstack{Blocking}\end{turn}} 
                   & \textit{SA} & Saad's reordering algorithm~\cite{Saad}
                   \\
& \textit{1-SA} & Our reordering algorithm, as detailed in Algorithm \ref{algo:sparta} \\
                   
& \textit{$\tau$} & Similarity threshold in SA and 1-SA\\
%& &\\
& \textit{h(v)} & The hash for vector v, evaluated as in Eq.~\ref{hash}
\\
\\
\hline

\end{tabular}
\caption{Notation used in this paper.}
\label{tab:symbols_conv}
\end{table}

\section{Methodology and Algorithms}
\label{sec: methods-and-algos}
In this section, we present our one-dimensional modification of SA, called 1-SA. Its objective is to find a partition of the rows that clusters the nonzero entries of the matrix into blocks.   %--

The basic idea is to partition the columns in advance, and then reformulate the hash function \eqref{hash} and the similarity function \eqref{jaccard_sim} to work on the "quotient" rows determined by the column partition. The advantage of this approach is that it is able to pair rows whose elements belong to spatially close columns. Moreover, it can control exactly the ordering and blocking of the second dimension by choosing the columns partition, which allows the blocking to be better tuned to the multiplication task at hand.
We present a detailed analysis of our algorithm, along with some variations thereof, in section~\ref{reorder-algo}. 

\subsection{The reordering and blocking algorithm}
\label{reorder-algo}

We now detail how to extend SA to obtain a 1-dimensional blocking algorithm that fits our needs. We will refer to this algorithm as 1-SA.
First, a partition $Q$ for the columns is chosen. From now on we will consider the case of a regular partition with blocks of width $\Delta_W$. The choice of $\Delta_W$ can be guided by knowledge about either the hardware or the matrix block structure. In the multiplication experiments of Section~\ref{sec: exp:multiplication}, we will always select $\Delta_W$ to be a multiple of 64, for the benefit of the Tensor Cores-based multiplication routine. 
Given a column partition $Q = {Q_1 .... Q_K}$, we consider the $Q-$quotient version of a row $v_i$, namely a K-dimensional binary vector $\hat{v}_i$ such that 
\begin{equation}
\hat{v}_{i,j} = 
     \begin{cases}
       \text{1} \quad \text{if there is a } v_{i,k} \neq 0 \; \text{ in } Q_j\; \\
       \text{0} \quad \text{otherwise} \\ 
     \end{cases}.
 \label{compressedrow}
 \end{equation}
That is, the $k$-th element of $\hat{v}$ is nonzero if and only if $v$ has at least a nonzero in the positions corresponding to $Q_k$. 

We then apply algorithm~\ref{algo:hash} to the quotient rows: for each quotient row $\hat{v}$ (as in~\ref{compressedrow}) a hash number is calculated as follows:
\begin{equation}
    \mathrm{h}(\hat{v}) = \sum_{\{i \mid \hat{v}_i \neq 0\}}{i}.
    \label{hashpartition}
\end{equation}
Note that $\mathrm{h}(\cdot)$ can be replaced by any hash function with stronger theoretical guarantees.
 After checking for collisions, identical (quotient) rows are grouped together. Note that two rows with different nonzero patterns may be grouped together, as long as their quotient pattern is the same. From now on, the rows have been \textit{compressed}. A compressed row may represent any number of identical (quotient) rows. As in SA, the compression step is not necessary, but it will speed-up the rest of the computations. Without loss of generality, from now on we will assume the rows of our CSR matrix to represent any number of identical rows.
 
After compressing identical rows, we create blocks by merging similar rows, in a fashion similar to SA. We sketch our blocking algorithm in Algorithm~\ref{algo:sparta}. 
The main loop in lines 4-18 starts by creating a group from the first not yet merged row (lines 7-8). Then, through the loop in lines 9-16, we iteratively try to merge compressed rows to the current group.
As in SA, we check (line 11) if the similarity of the row with the current group exceeds a certain threshold $\tau$. Unlike SA, we evaluate a more general \textit{merge condition} before deciding whether to merge a row to the current group. In Section~\ref{sec:bounds} we show that some theoretical bounds on reorder quality can be enforced if the merging criterion involves a minimum Jaccard similarity with an additional control on how a group of rows can grow.

Finally, when a row is merged into a group (line 12), we update the block-pattern of the group to the bitwise-OR of the two merged patterns (line 13). The motivation behind this choice is that, from the moment a nonzero element is added to a group, all zeros in the corresponding group of columns will be treated as nonzeros.
This loop continues until each row has been assigned to a group. The output of the algorithm is a partition of the rows in groups. Together with the columns partition Q, this determines a blocking of the matrix. Converting the matrix to the VBR formats using these rows and columns partitions allows for saving only the nonzero blocks of the matrix during the rest of the procedure. We provide a visual representation of 1-SA in Fig \ref{fig:algo_explain}. 

\begin{figure}[h]
    \centering
    \begin{subfigure}[t]{0.2\textwidth}
        \centering
        \includegraphics[width = \linewidth,trim={11cm 3.6cm 11cm 1cm},clip]{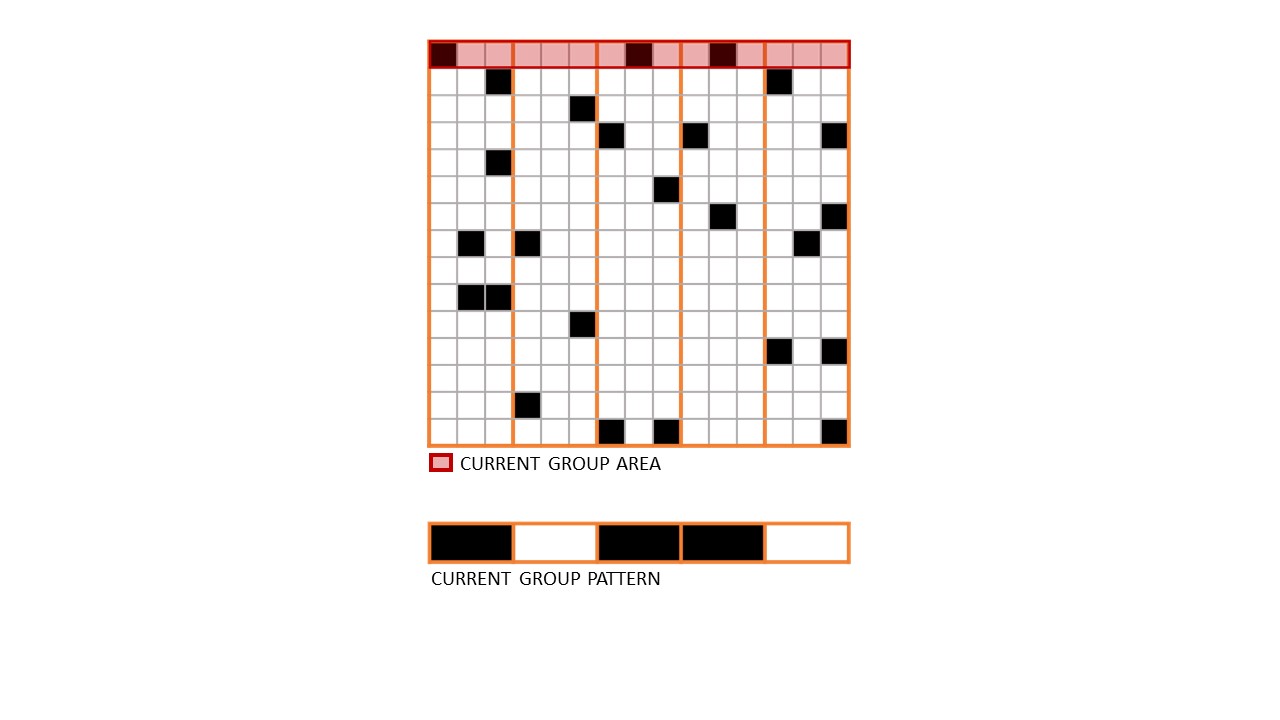}
        \caption{}
    \end{subfigure}%
    \begin{subfigure}[t]{0.2\textwidth}
        \centering
        \includegraphics[width = \linewidth,trim={11cm 3.6cm 11cm 1cm},clip]{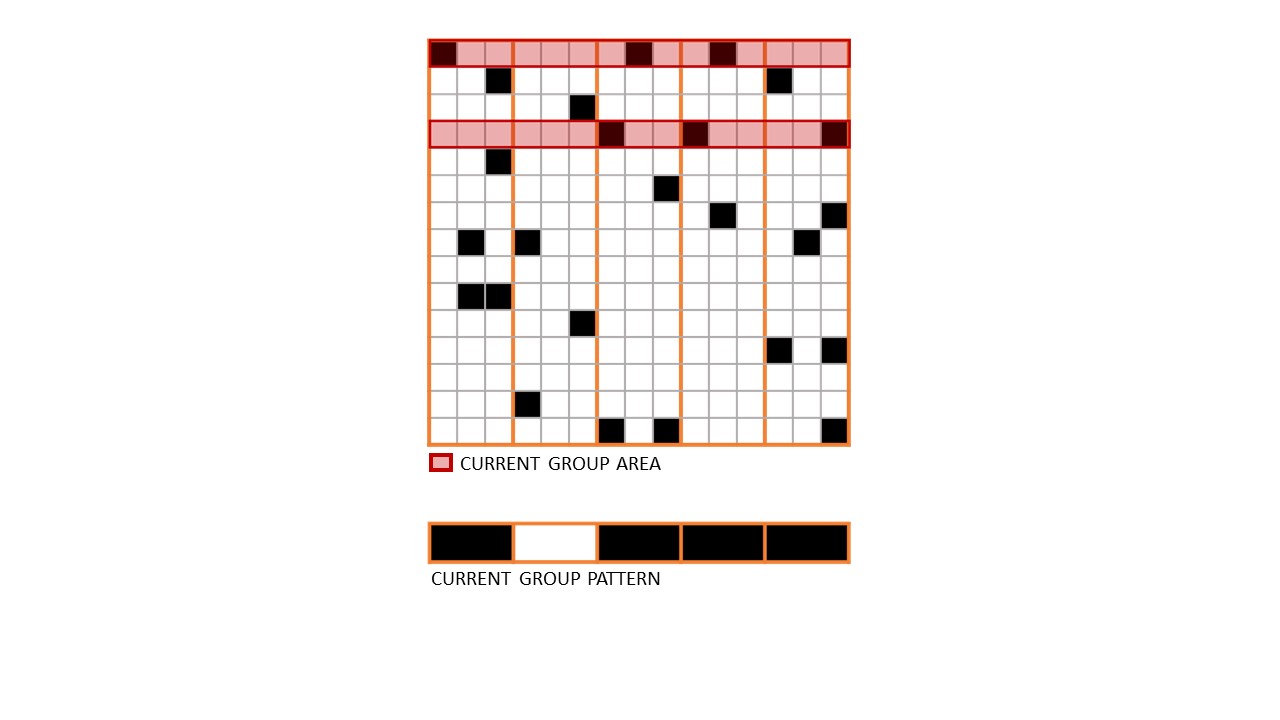}
        \caption{}
    \end{subfigure}
    \begin{subfigure}[t]{0.2\textwidth}
        \centering
        \includegraphics[width = \linewidth,trim={11cm 3.5cm 11cm 0},clip]{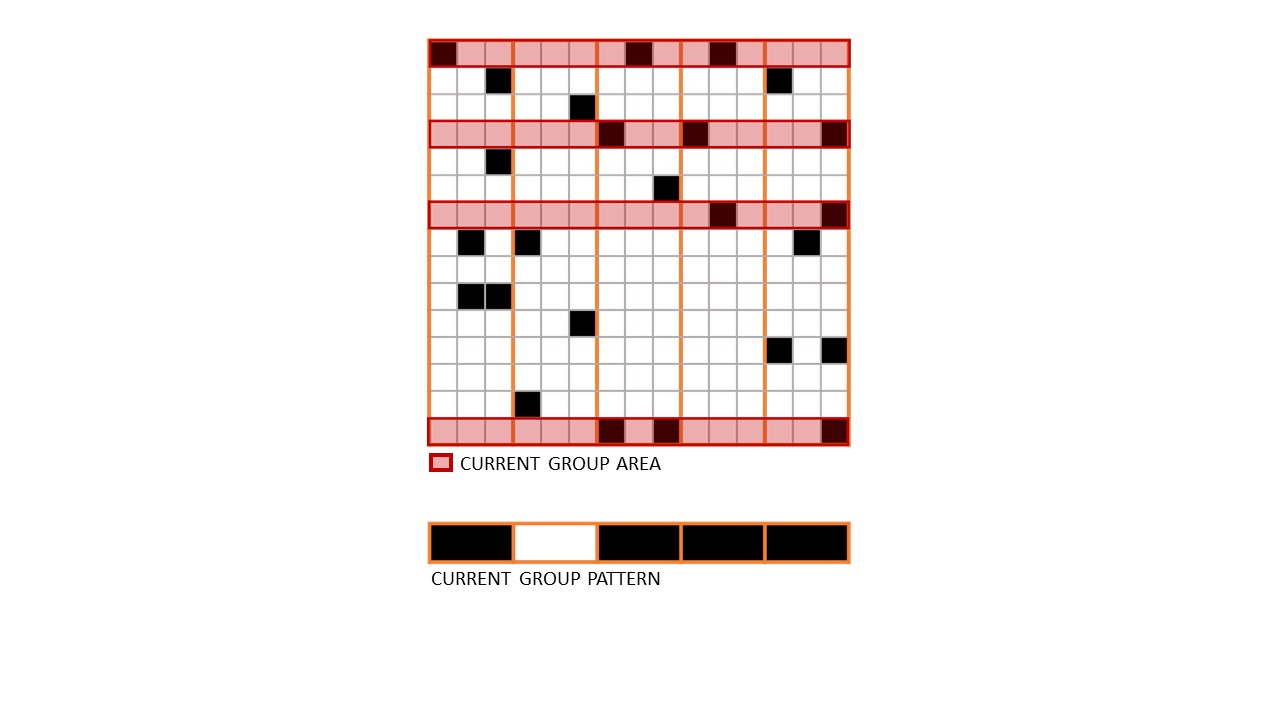}
        \caption{}
    \end{subfigure}
    \begin{subfigure}[t]{0.2\textwidth}
        \centering
        \includegraphics[width = \linewidth,trim={11cm 3.5cm 11cm 0},clip]{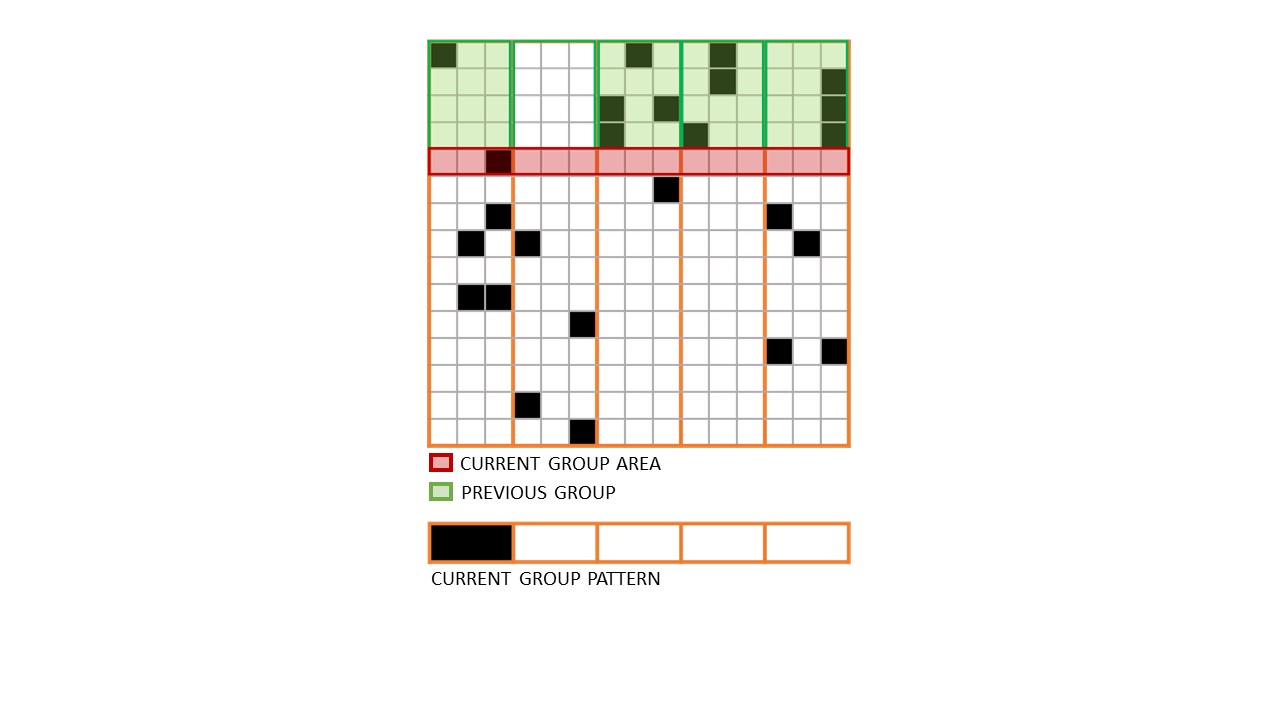}
        \caption{}
        \label{subfig: saad-points}
    \end{subfigure}

    \caption{Visualisation of the 1-SA blocking algorithm on a $ 15\times 15$ sparse matrix with $\Delta_W = 3$ and $\tau = 0.5$. In (a) the first row is selected to create the first block (red highlight). The current group pattern is thus set to the projection of the first row. Subsequent rows will be compared to this pattern, and merged if they satisfy the merging criterion. In (b), the first suitable row is discovered. It is merged to the current group, and the pattern is updated accordingly. The algorithm continues until all rows have been evaluated and possibly merged (c). Then (d) a new group is created from the next merged row, and process repeats. Note that already grouped rows (green highlight) will not be evaluated again, and will create a row of blocks in the final matrix. This fact is visualized by partially reordering the rows in (d). In practice, the rows will be only reordered when all groups have been created.}
    \label{fig:algo_explain}
\end{figure}

In summary, we propose the following modifications to obtain 1-SA from SA:
\begin{itemize}
    \item a column partitioning, which allows the algorithm to recognize rows with spatially close entries;
    \item a pattern update after each merge, which allows for better similarity evaluations at each step;
    \item a two-steps merging criterion, which ensures theoretical bounds on the final density.

\end{itemize}

\begin{algorithm}
	\caption{1-SA algorithm}
		Input: A CSR matrix, a merge condition function. \\
		Output: a mapping from rows to groups
	\begin{algorithmic}[1]
	\For{i = $1,\ldots, N$}
		    \State $\mathrm{group}[i] = -1>$
	\EndFor
	\For{$i = 1,\ldots,N$}
	\If{$\mathrm{group}[i] \neq -1$} \textbf{skip}
	\EndIf
	\State group$[i] = i$;
	\State ThisPattern = $\hat{v}_i$
	    \For{$j = i+1, \ldots, N$}
	        \If {$\mathrm{group}[j] = -1$}
    	        \If {MergeCondition(thisPattern, $\hat{v}_j$)}
    	            \State $\mathrm{group}[j] = i$
    	            \State thisPattern = OR(thisPattern, $\hat{v}_j$)
    	        \EndIf
	        \EndIf
	    \EndFor
	\EndFor
	\end{algorithmic} 
	\label{algo:sparta}
\end{algorithm}

\subsection{Reordering with theoretical bounds}
\label{sec:bounds}

Merging rows while the similarity between the pattern and the new row is at least $\tau$ often suffices for constructing sufficiently dense blocks from an experimental point of view. 
However, there can be pathological matrices that can generate sparse blocks with densities arbitrarily close to 0.
Consider the following example. Let $\ell>0$ be an arbitrary value and let $\tau\geq 0.5$ the Jaccard similarity threshold. Consider the following set of $\ell+\ell^{1/4}$ rows $v_0, \ldots v_{\ell+\ell^{1/4}-1}$: row $v_i$ for $i\in[0,\ell)$ has one nonzero entry in the first column and zeros otherwise; row $v_{\ell+j}$ for $j\in [0,\ell^{1/4})$ has nonzero entries in the first $j$ columns and zero otherwise. 
By merging rows in the order $v_0, \ldots v_{\ell+\ell^{1/4}-1}$,  we get that the similarity is 1 up to row $v_{\ell-1}$ and then it is in the range $[0.5,1)$. For $\tau\geq 0.5$, we then get a block of size $(\ell+\ell^{1/4})\times (\ell^{1/4})$ with $\ell+\ell^{1/4}(\ell^{1/4}+1)/2$ nonzero entries: the density is then $\BT{1/\ell^{1/4}}$. 

We now propose a stronger merging condition with theoretical guarantees on the density of the output blocks.
Consider a group of rows with pattern $p$ and let $v_0$ be the first row added to the group; then a new row $v$ is added to the group if: 
\begin{enumerate}
\item The similarity between the new row $v$ and the pattern $p$ is above the threshold $\tau$; 
\item The number $\lambda$ of nonzeros in the new pattern $p'=OR(p,v)$ is at most $\lambda_0/(1-\tau/2)$, where $\lambda_0$ is the number of nonzeros in $v_0$.
\end{enumerate}
We get the following result for Jaccard similarity.

\begin{theorem}\label{thm:density}
Let $G$ be a group of rows $v_0,\ldots v_{h-1}$, merged by the algorithm in this order using the aforementioned merging condition with Jaccard similarity threshold $\tau$. Let $\lambda_0$ be the number of nonzeros in $v_0$, and let $\lambda$ be the total number of nonzero in the final pattern.
If $\lambda \leq \lambda_0/(1-\tau/2)$, then the density $\rho_G$ of $G$ after removing empty columns is at least $\tau/(2\Delta_W)$.
\end{theorem}
\begin{proof}
We initially assume $\Delta_W=1$.
Consider a group $G$ of rows after the reordering, and let $\rho_G$ be its density after removing all empty columns. 
For the sake of notational simplicity, let us assume that it consists of $h$ rows $v_0,\ldots v_{h-1}$, added by the algorithm in this order. 
For $0\leq i < h$, let $p_i$ be the pattern after adding row $v_i$,  that is $p_i$ is the element-wise OR of $v_0,\ldots v_{i}$; we also  let $V_i$ and $P_i$ represent the set containing column indexes of nonzero entries in $v_i$ and $p_i$  (for convenience, we set $P_{-1}=\emptyset$).
For $i\geq 0$, we define $\Lambda_i = V_i \setminus P_{i-1}$ and $\Gamma_i = P_{i-1}  \setminus V_i$.
The set $\Lambda_i$ contains the column indexes that are added in $G$ by row $v_i$: for each $k\in \Lambda_i$, we have that column $k$ contains only zero entries  in row $v_j$ for any $j<i$.
On the other hand, the set $\Gamma_i$ contains the column indexes corresponding to zero entries in $v_i$, but for which there exists at least a row in $v_0,\ldots v_{i-1}$ with a nonzero entry.
We define  $\gamma_i=|\Gamma_i|$ and $\lambda_i=|\Lambda_i|$, and $\Lambda=P_{h-1}$ (i.e., the set of all columns in $G$ with at least a nonzero entry).
Note that $\sum_{i=0}^{h-1}\lambda_i  = \lambda$.

We observe that the zero entries in row $v_i$ must be either in empty columns (which are removed in the density estimate), in columns in $\Gamma_i$ (i.e., the columns with nonzero entries in the preceding $i-1$ rows, but not in $v_i$) or in $\Lambda\setminus P_i$ (i.e., the columns with nonzero entries in the rows that will be added by subsequent rows but not in the first $i$ rows). Then, the total number $Z$ of zeros in $G$ after removing all empty columns is:
\begin{align*}
Z & = \sum_{i=0}^{h-1} \left(|\Gamma_i |+ | \Lambda\setminus P_i |\right) \leq \sum_{i=0}^{h-1} \left(\gamma_i + \lambda  -\lambda_0\right)\\
&\leq (\lambda-\lambda_0) h + \sum_{i=0}^{h-1} \gamma_i.
\end{align*}
By construction, we add row $v_i$ if $Jaccard(P_{i-1}, V_i)\geq \tau$ and hence 
$$\gamma_i  = |P_{i-1} \setminus  V_i| \leq |V_i\setminus P_{i-1}| +  |P_{i-1} \setminus V_i | \leq (1-\tau) |P_{i-1} \cup V_i|.$$
Since $|P_{i-1} \cup V_i| = \sum_{j=0}^{i-1} |\Lambda_j| = \sum_{j=0}^{i-1} \lambda_j$, we get:
\begin{align*}
Z & \leq  h(\lambda-\lambda_0) +  (1-\tau) \sum_{i=0}^{h-1}  \sum_{j=0}^{i-1} \lambda_j \\
& \leq  h(\lambda-\lambda_0) +  (1-\tau) h \sum_{j=0}^{h-1}   \lambda_j \\
&\leq  h(\lambda-\lambda_0) + (1-\tau)h \lambda \\
&\leq h\left((2-\tau) \lambda- \lambda_0\right).
\end{align*}
Since the total number of entries in $G$ is $h \lambda$, the density $\rho_G$ is then
$$
\rho_G \geq 1-\frac{\left((2-\tau)\lambda - \lambda_0\right)}{\lambda}.
$$
Since $\lambda \leq   \lambda_0/(1-\tau/2)$,  we get $\rho_G \geq \tau/2$.

For the case $\Delta_W>1$, we perform the previous analysis on the quotient rows. 
Since each nonzero entry in the compressed matrix corresponds to at least one nonzero and at most $\Delta_W-1$ zeros, we get the claimed result. 

\end{proof}

\subsection{Analysis of algorithms}
\subsubsection{Cost of 1-SA}
We consider a sparse matrix stored in CSR format. Without loss of generality, we ignore the actual values of the entries and only focus on the structure of the nonzero elements (i.e., we replace the nonzero values with 1).

First, we analyze Algorithm~\ref{algo:sparta} ignoring the pattern update (step 13).
The running time of this simplified algorithm is $O(N^2k)$ for a $N\times N$ sparse matrix with, at most, $k$ nonzero elements per row. The worst case happens when no two rows are merged, so that $N(N-1)/2$ row comparisons have to be carried on. For both Jaccard and cosine similarity, comparing row $v_i$ with row $v_j$ takes $O(k_i + k_j)$ operations. Comparing row $v_i$ with every other row $v_j$ with $j > i$ takes $O((N-i)k_i + k_{i+1} +... + k_N)$ steps, from which the overall $O(N^2k)$ bound follows.

If the pattern update (step 13) is used, the pattern may grow during a single pass through steps 4-17. When a row $v_j$ is merged to the current pattern, the latter grows at most by $k_j$ entries, so that the cost of the remaining comparisons becomes $(N-j)(k + k_j) + k_{j+1} +... + k_N$ assuming no other row is merged. Yet, since row $v_j$ has now been merged, the algorithm will not try to build a group out of it (see step 5), which would have a cost proportional to $(N-j)(k_j) + k_{j+1} +... + k_N$. This is exactly the additional cost incurred by the pattern update. We conclude that switching on the pattern update does not influence the worst-case time analysis.

\subsubsection{Cost of matrix multiplication with reordering}
In this section, we provide an upper bound on the cost of multiplying a sparse matrix reordered with the strategy in 
Section~\ref{sec:bounds} with a dense matrix. 
For the analysis, we use the $(m,\ell)$-TCU model studied in~\cite{CSV21,AS20}, where $m$ and $\ell$ are two 
parameters characterizing the tensor. 
Specifically, the $(m,\ell)$-TCU  model is a traditional RAM model featuring a tensor core unit that can natively compute the matrix multiplication between two dense $\sqrt{m}\times \sqrt{m}$ matrices in time $\BO{m+\ell}$, where $\ell$ is a latency cost.
The multiplication between a  $r\times c$ matrix and a $c \times s$ matrix requires time $\BO{r c s / m^{1/2} + cs \ell / m}$ on a $(m,\ell)$-TCU.

Consider a $N\times N$ matrix $A$ with $K$ nonzero entries. Assume that the reordering described in the previous Section~\ref{sec:bounds} gives $H$ groups and let the $i$-th group $G_i$ be a $r_i\times c_i$ matrix.
Let $B$ be a dense $N\times N$ matrix. 
We now provide an upper bound on the running time for computing $A\cdot B$ using the $(m,\ell)$-TCU: we multiply each group $G_i$ with $B$ by using the aforementioned dense matrix multiplication algorithm~\cite{CSV21}.
For simplicity, we assume $\Delta_W=1$.

\begin{theorem}
Let $A$ be a sparse $N\times N$ matrix with $K$ nonzero entries. Assume that the reordering  procedure in Section~\ref{sec:bounds}, with Jaccard similarity threshold $\tau$ and $\Delta_W=1$, gives $H$ blocks of size $r_i\times c_i$ with $i\in\{0,\ldots H-1\}$. 
Let $B$ be a dense  $N\times N$ matrix $B$.
If $r_i \geq \sqrt{m}$ for a constant fraction of the $H$ blocks, then it is possible to compute $A\cdot B$ in $$\BO{\frac{K N}{m^{1/2} \tau} + \frac{K N \ell}{m^{3/2}\tau}}$$ on the $(m,\ell)$-TCU model.
\end{theorem}
\begin{proof}
A constant fraction of the $H$ blocks have more than $\sqrt{m}$ rows and, by Theorem~\ref{thm:density}, their density is at least $\tau/2$ after removing the empty columns. 
We can then pad the remaining blocks to reach $\sqrt{m}$ rows and density at least $\tau/2$ without asymptotically increasing the running time.
We then assume that for all blocks $r_i\geq \sqrt{m}$ and the density is at least $\tau/2$.

Consider a block $G_i$ given by the reordering with $r_i$ rows and $c_i$ non empty columns.
By representing the group as a dense $r_i\times c_i$ matrix and by extracting the corresponding submatrix $B$ of size $c_i\times N$ in $\BO{c_i n}$ time, we can compute the product between the $G_i$ and the dense matrix in $\BO{r_i c_i N/m^{1/2} + c_i N \ell / m}$ time.
Since the density of $G_i$ is at least $\tau/2$, we get $K_i / (r_ic_i)\geq \tau/2$ where $K_i$ is the number of nonzeros in $G_i$. It then follows that $r_ic_i\leq 2 K_i/\tau$ and that $c_i\leq 2K_i/(\sqrt{m} \tau)$  since $r_i\geq \sqrt{m}$.

Therefore, the multiplication between $A$ and $B$ can be computed in 
\begin{align*}
T &= \BO{\sum_{i=0}^{H-1} \frac{r_i c_i N}{m^{1/2}} + \frac{c_i N \ell}{m}}= \BO{\sum_{i=0}^{H-1} \frac{K_i N}{m^{1/2} \tau} + \frac{K_i N \ell}{m^{3/2}\tau}}\\
&= \BO{\frac{K N  }{m^{1/2} \tau} + \frac{K N \ell}{m^{3/2}\tau}}. 
\end{align*}
\end{proof}
We observe that, when $\ell=\BO{m}$ and $\tau=\BT{1}$, we can compute $A\cdot B$ in time $\BO{KN/{m^{1/2}}}$, which is a factor $m^{1/2}$ faster than the trivial algorithm.

%\section{Revising a TCU for Sparse Computation}

\section{Experiments}
\label{sec: exp}
One of the goal of this study is to evaluate the quality of our reordering and blocking algorithm. Since it is computationally infeasible to compute the optimal blocking for a generic sparse matrix (even of modest size), in Section~\ref{sec: exp:reorder} we generate synthetic matrices with a known blocking structure, scramble the order of their rows, and evaluate the ability of 1-SA to retrieve the original blocking.
With this method, we study the effectiveness of 1-SA across the landscape of blocked matrices, evaluating it on a wide range of sparsity patterns.

Then, in Section~\ref{sec: exp:multiplication}, we study the performance improvements when multiplying blocked matrices with dense ones. To this end, we evaluate the performance of a blocked multiplication kernel built using routines implemented in cuBLAS library, and compare it with the performance of the sparse-specific multiplication routine from cuSPARSE. We also validate our methodology on synthetic RMATs~\cite{rmat} and on real-world sparse matrices from the Network Repository~\cite{netref} collection.

The code is publicly available on the github repository at the following url: \url{https://github.com/LACSFUB/SPARTA.git} 

\subsection{Dataset description and generation}
\label{sec: synthetic-generation}
\label{dataset-description}
For the experiments of section~\ref{sec: exp:reorder} we generated synthetic blocked sparse matrices with varying sparsity pattern. A matrix $A(\Delta, \theta, \rho)$ is generated as follows: first, we divide the matrix in $\Delta \times \Delta$ blocks. Then, we randomly select a fraction $\theta$ of these blocks that we flag as nonzero. Finally, for each nonzero block, we select a fraction $\rho$ of entries to be nonzero. We note that, following this procedure, both the number of nonzero blocks in a blocked row and the number of nonzero entries in a row are not fixed within the same matrix.
The parameters used to generate the matrices are detailed in Table~\ref{tab:synthdata}. 

For the experiments of Section~\ref{sec: exp:multiplication}, we used synthetic RMATs with generating parameters $(0.57,0.19,0.19,0.05)$. Their size and average degree are detailed in Table~\ref{tab:synthdata}.

We also evaluated our methods on real-world matrices from the Network Repository collection. Their characteristics are described in Table~\ref{tab:realdata}.

\begin{table}

\begin{tabular}{ |p{3.5cm}||p{4cm}|}
 \hline
 \multicolumn{2}{|c|}{Synthetic blocked matrices} \\
 \hline
 Parameter & Values \\
 \hline
 Rows, columns $(M,N)$        & $2^{11},2^{12},2^{13}$ \\
 Block size $\Delta$ & $2^5,2^6,2^7,2^8$\\
 Block density $\theta$   & $0.01, 0.1, 0.2, 0.3, 0.4$ \\
 In-block density $\rho$     & $0.01, 0.02, 0.05, 0.1, 0.2, 0.5$ \\
 \hline
 \multicolumn{2}{|c|}{Synthetic RMATs} \\
 \hline
 Parameter & Values \\
 \hline
 Nodes & $2^{14}$ \\
 Avg.Degree & $8,16,32,64,128$\\
 \hline
\end{tabular}
\caption{Parameters used for the generation of the synthetic block matrices and RMATs used in the experiments.}
\label{tab:synthdata}
\end{table}
\begin{table}
\begin{tabular}[ht]{ |p{2.75cm}||p{1.3cm}|p{1.3cm}|p{1.3cm}|}
\hline
\multicolumn{4}{|c|}{Graphs from the Network Repository collection\cite{netref}} \\
\hline
Graph & Nodes & Edges & Density\\
\hline
econ-mbeacxc &
493 &
49920 &
20.539\% \\
econ-beaflw &
498 &
53403 &
21.533\% \\
econ-beacxc &
498 &
50409 &
20.326\% \\
C500-9 &
501 &
112332 &
44.754\% \\
hamming10-4 &
1025 &
434176 &
41.325\% \\
bn-mouse-retina &
1112 &
577350 &
46.691\% \\
p-hat1500-3 &
1501 &
847244 &
37.605\% \\
bio-CE-PG &
1870 &
47754 &
1.366\% \\
fb-messages &
1900 &
61734 &
1.710\% \\
C2000-5 &
2001 &
999836 &
24.971\% \\
bio-SC-HT &
2084 &
63027 &
1.451\% \\
bio-CE-GN &
2218 &
53683 &
1.091\% \\
econ-orani678 &
2530 &
90158 &
1.409\% \\
qc2534 &
2535 &
232947 &
3.625\% \\
econ-psmigr2 &
3141 &
540022 &
5.474\% \\
econ-psmigr1 &
3141 &
543162 &
5.505\% \\
bio-DR-CX &
3287 &
84940 &
0.786\% \\
bio-DM-CX &
4039 &
76717 &
0.470\% \\
bio-HS-LC &
4226 &
39484 &
0.221\% \\
bio-HS-CX &
4413 &
108818 &
0.559\% \\
bio-grid-yeast &
6008 &
313890 &
0.870\% \\
nemeth24 &
9507 &
758028 &
0.839\% \\
ted-AB &
10606 &
522387 &
0.464\% \\
bio-CE-CX &
15229 &
245952 &
0.106\% \\
bio-WormNet-v3 &
16347 &
762822 &
0.285\% \\
Si10H16 &
17078 &
446500 &
0.153\% \\
ca-AstroPh &
17904 &
196972 &
0.061\% \\
ia-retweet-pol &
18469 &
61157 &
0.018\% \\
ia-wikiquote &
21608 &
549210 &
0.118\% \\
movielens-10m &
28139 &
286740 &
0.036\% \\
SiO &
33402 &
675528 &
0.061\% \\
bmw7st-1 &
141348 &
514273 &
0.003\% \\
ca-dblp-2010 &
226414 &
716460 &
0.001\% \\
 \hline
\end{tabular}
\caption{Description of the real-world sparse matrices from the Network Repository graph collection~\cite{netref}.}
\label{tab:realdata}
\end{table}

\subsection{Experimental Setup}
We ran all experiments in Section~\ref{sec: exp:multiplication} on a Nvidia V100 GPU with 32GB of memory. We used CUDA 11.4.0 and the corresponding versions of cuBLAS and cuSPARSE.

\subsection{Reordering}
\label{sec: exp:reorder}

\subsubsection{Synthetic block matrices --- Blocking curves}
We evaluated the efficacy of our blocking algorithm on synthetic block-sparse matrices. A block-sparse matrix $A(\Delta,\theta,\rho)$ is generated by fixing the height and width of blocks $\Delta_W = \Delta_H = \Delta$, the fraction of nonzero blocks $\theta$ and the in-block density $\rho$, as detailed in~\ref{sec: synthetic-generation}.
A reordering experiment on the synthetic block matrix $A$ is performed as follows: first, the rows of the matrix are scrambled. Then, the columns are partitioned uniformly with width $\Delta$. Finally, for a given $\tau$, the blocking algorithm is run, producing the reordered and blocked matrix $A'_{\tau}$ (the subscript will be dropped from now on). 

We investigate the blocking curves (see Section~\ref{sec:blocking-quality}) produced by varying $\tau$ in our algorithm. In these curves, we observe the trade-off between the in-block density $\rho'$ against the block height $\Delta'_H$. 

Note that, since blocks in $A'$ are uneven in size and number of elements, $\rho'$ is to be intended as the overall density of nonzero elements in the nonzero area of $A'$, and $\Delta'_H$ as the average height of nonzero blocks in $A'$.

The reordering curves in Figure~\ref{fig:blocking_curves} show how our blocking algorithm can explore the trade-off between $\Delta'_H$ and $\rho'$. Each point corresponds to a different value of $\tau$, and we evaluate the density $\rho'$ at $\Delta'_H \approx \Delta$ to gauge the ability of our algorithm to recover the original blocking. In the rest of the curve, we observe how using lower (or higher) $\tau$ can trade in-block density to increase the block-height (or vice-versa).
\begin{figure}[htb!]
    \centering
    \includegraphics[width = 0.8 \linewidth]{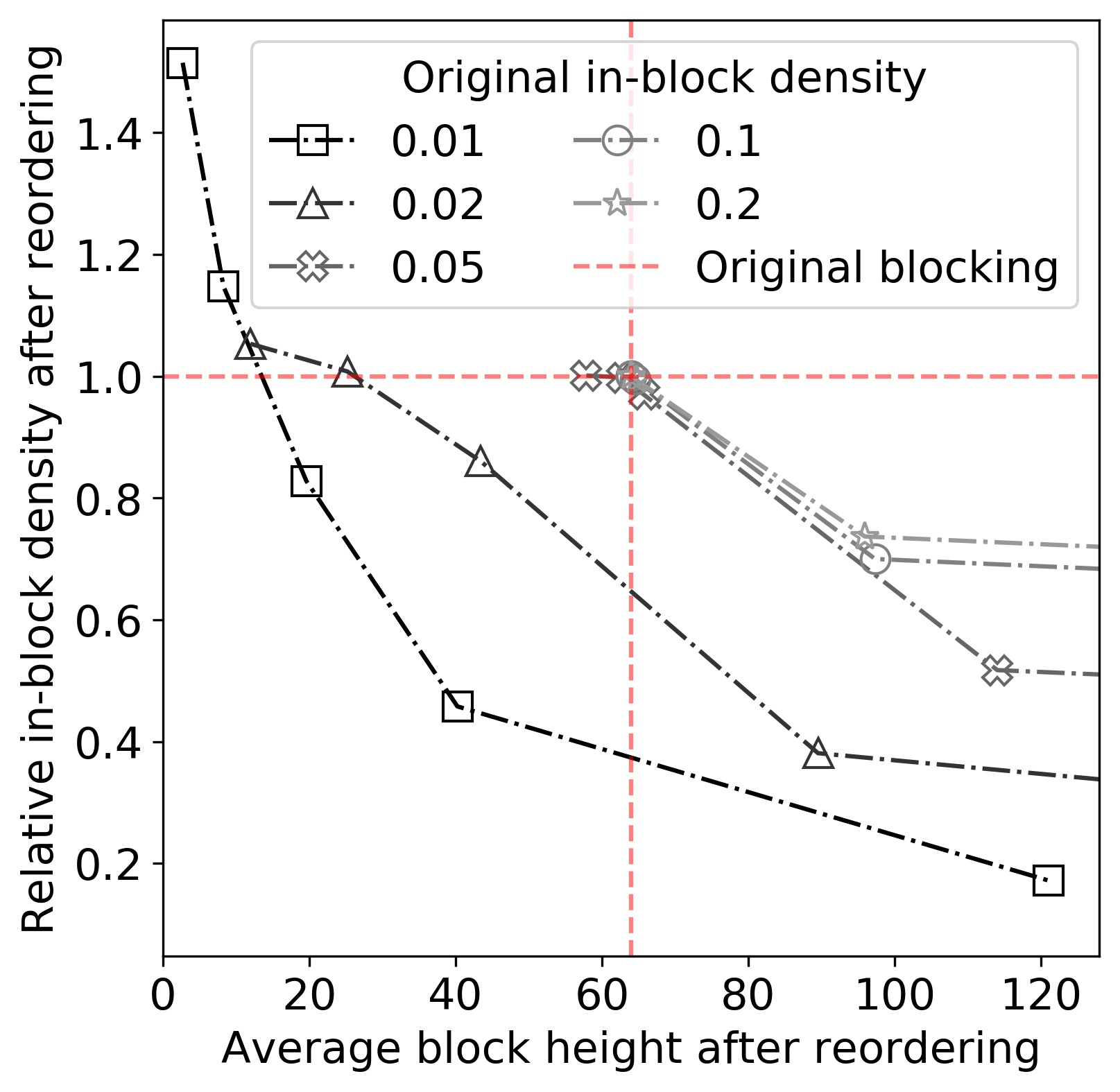}
    \caption{Blocking curves for 1-SA on five synthetic blocked matrices, showing the blocking trade-off between block size (x-axis) and in-block density (y-axis). The $8192 \times 8192$ matrices are divided in blocks of size $64 \times 64$, with $10\%$ randomly chosen nonzero blocks. Matrices differ by in-block density (see the legend). Notice how 1-SA can recover the original blocking (red cross) for the three denser matrices, and obtain around half of the optimal density for the sparser ones.}
    \label{fig:blocking_curves}
\end{figure}

\subsubsection{Synthetic block matrices --- landscape}
To better gauge the overall efficacy of our blocking algorithms, we tested it extensively over the landscape of synthetic blocked sparse matrices. That is, we generated blocked sparse matrices from a range of $\theta$, $\rho$ and $\Delta$ values, we scrambled the order of their rows, and then we blocked them with $\tau = [0.1, 0.2, ... ,1]$ . We then selected a blocking with $\Delta'_H \approx \Delta$, and recorded the relative in-block density $\frac{\rho'}{\rho}$. Similarly, we recorded the block height $\Delta'_H$ at $\rho' \approx \rho$. These values provide a measure of how good the algorithm recovered the original blocking. A perfect recover occurs when $\rho'=\rho$ and $\Delta'_H = \Delta$. The main findings of these experiments are reported in Figure~\ref{fig:reorder-landscape}, showing how even unstructured and sparse matrices (high $\theta$, low $\rho$) can be blocked effectively by 1-SA.

\begin{figure}
    \centering
    \begin{subfigure}[t]{0.4\textwidth}
        \centering
        \includegraphics[width = \linewidth]{ 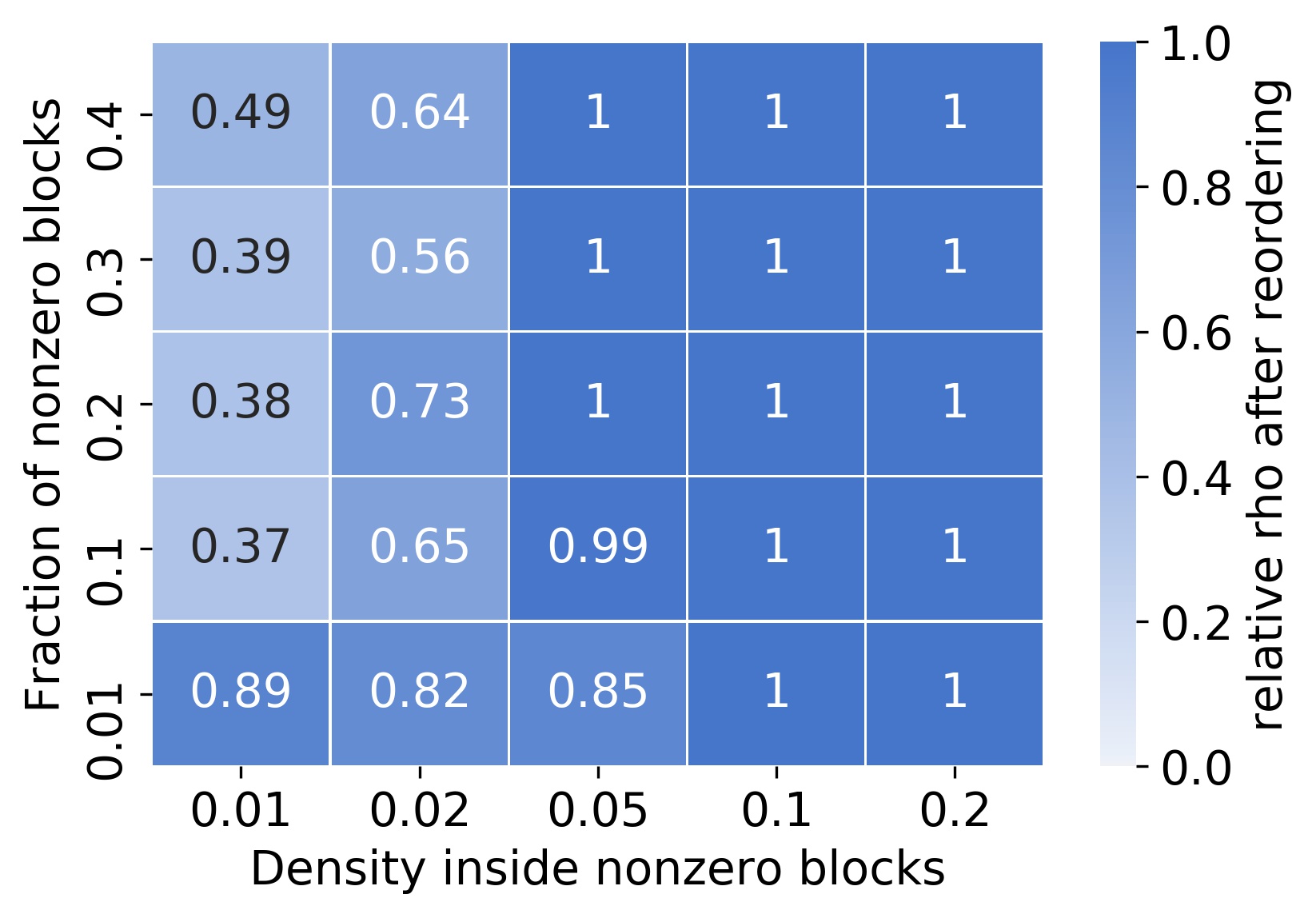}
        \caption{}
        \label{landscape-rho}
    \end{subfigure}%
    \hfill
    \begin{subfigure}[t]{0.4\textwidth}
        \centering
        \includegraphics[width = \linewidth]{ 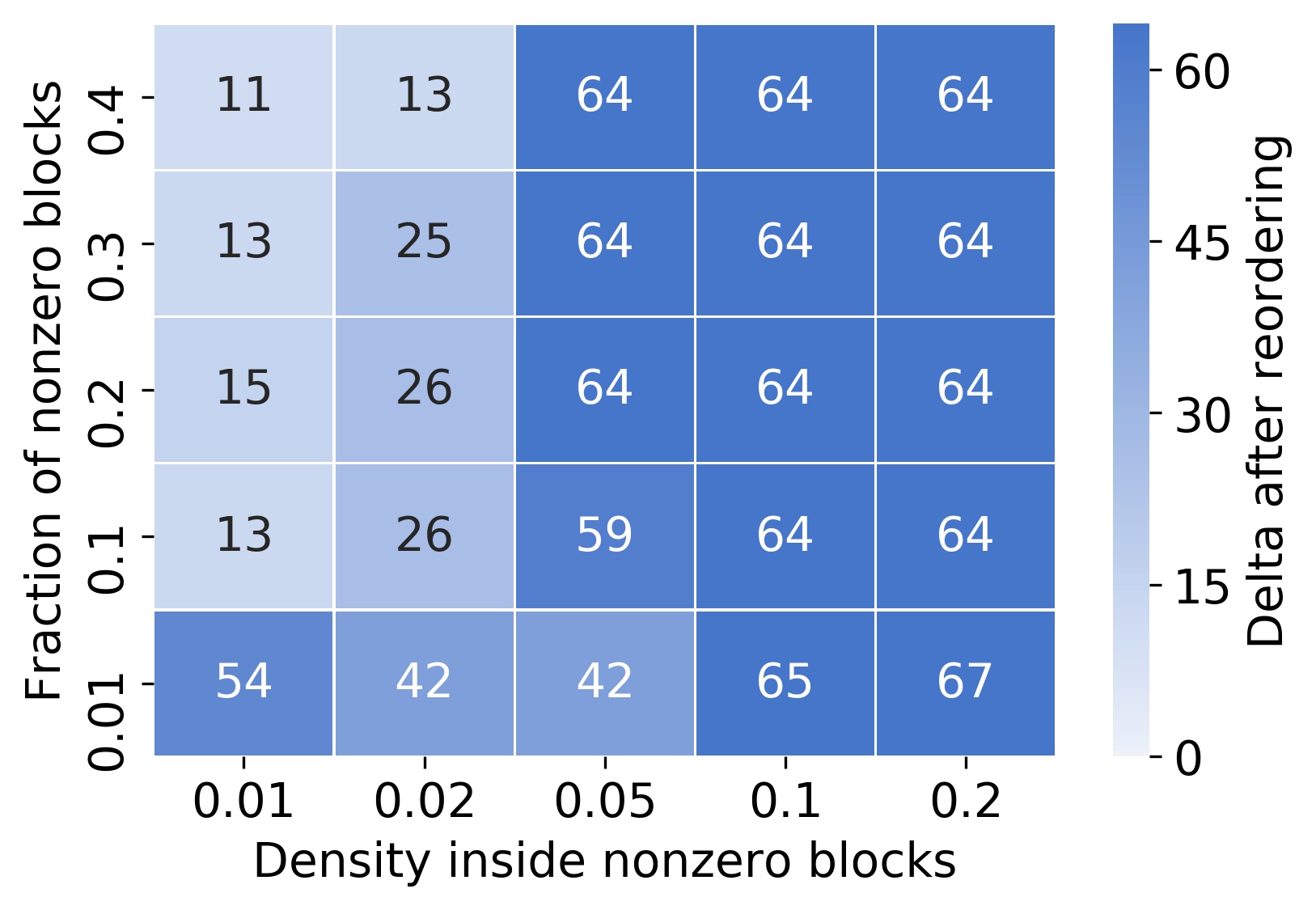}
        \caption{}
        \label{landscape-delta}
    \end{subfigure}

    \caption{Blocking quality over the landscape of synthetic $8192 \times 8192$ blocked sparse matrices. The x- and y-axis represent respectively the in-block density $\rho$ and the fraction of nonzero blocks $\theta$ of the original matrix. Values in (a) represent the post-reordering (relative) in-block density $\frac{\rho'}{\rho}$ at the original block size $\Delta'_H \approx 64$, with a value of $1$ meaning that a blocking equivalent to the original one has been found. Values in (b) represent the post-reordering block height $\Delta'_H$ at the original density $\rho' \approx \rho$.}
    \label{fig:reorder-landscape}
\end{figure}

\subsubsection{Naive implementation of SA}
Figure~\ref{fig:saad_comparison} shows the effectiveness of blocking using the naive implementation of SA (that is, without using the projected rows, the hierarchical merging, and the merge limit). As can be readily observed in the images, SA is unable to capture the block structure as effectively as our algorithm. This is because SA cannot recognize that two entries belong to the same block unless they share the same column index. 
\begin{figure}[h!]
    \centering
    \begin{subfigure}[t]{0.4\textwidth}
        \centering
        \includegraphics[width = \linewidth]{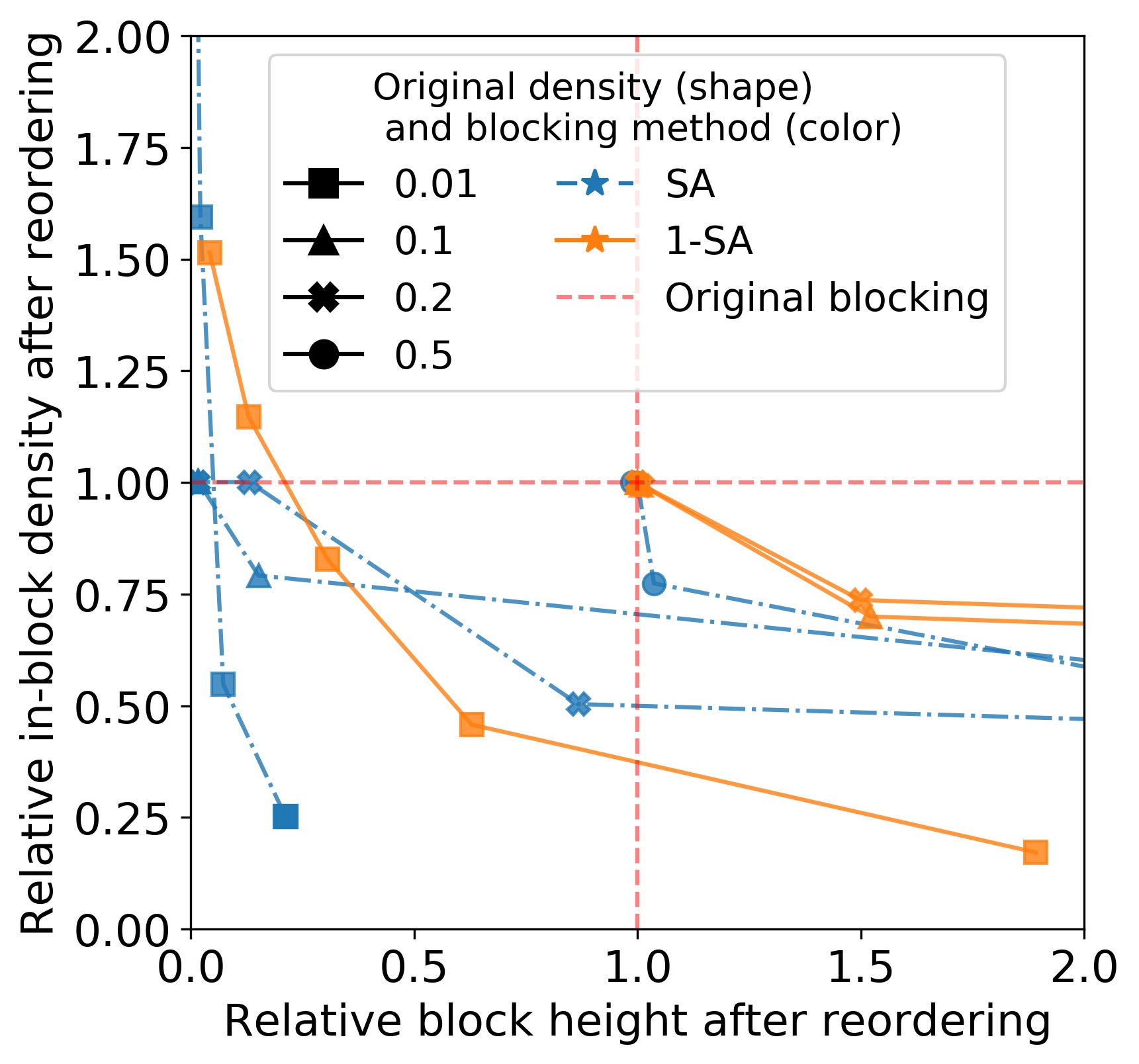}
        \caption{}
        \label{subfig: saad-comp-1}
    \end{subfigure}%
    \hfill
    \begin{subfigure}[t]{0.4\textwidth}
        \centering
        \includegraphics[width = \linewidth]{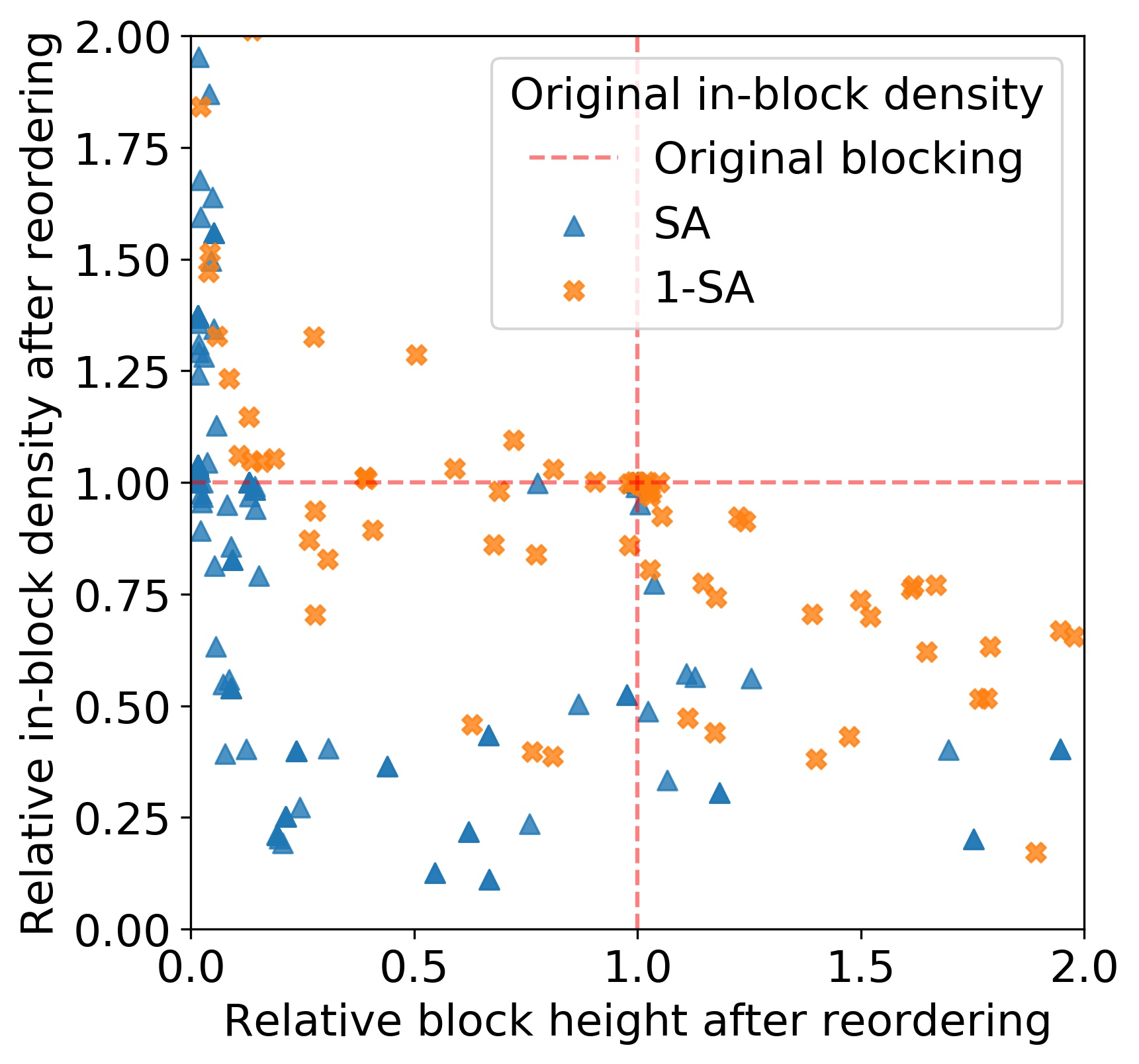}
        \caption{}
        \label{subfig: saad-comp-2}
    \end{subfigure}

    \caption{Blocking curves for 1-SA and SA. The x- and y-axis represent respectively the post-blocking relative density $\frac{\rho'}{\rho}$ and relative block height $\frac{\Delta_H'}{\Delta}$. Several blockings are obtained for each matrix by varying $\tau$. (a) shows the blocking curves for five $8192 \times 8192$ synthetic block matrices with $\Delta = 64$, and $\theta = 0.1$. (b) shows all blocking points for all the processed synthetic matrices (see Sec. \ref{dataset-description}). It can be seen that 1-SA can generate better (i.e. with higher $\rho'$ and $\Delta'_H$) blockings.}
    \label{fig:saad_comparison}
\end{figure}

\subsection{Reordering and Multiplication}
\label{sec: exp:multiplication}

\subsubsection{The multiplication routine}
We have implemented a simple, block-based, sparse-by-dense multiplication routine for matrices stored in the VBR format by resorting to the cuBLAS library. The routine uses cuBLAS streams to process
rows of blocks in parallel and uses the cuBLAS gemm routine to multiply the blocks using tensor cores. For the sake of brevity, we will refer to this block-based multiplication as VBR-cuBLAS. We note that our method is not restricted to the cuBLAS library, and can be readily adapted to any high-level parallel multiplication routine for dense matrices. As shown in the following experiments this routine, applied to the blocking generated by 1-SA, achieves considerable speed-ups against the cuSPARSE baseline. 

\subsubsection{1-SA vs cuSPARSE}
 In this experiment, we determine for which multiplication instances it is convenient to employ a blocking-based multiplication instead of a sparse-based one. We benchmarked the block-based cuBLAS multiplication routine on matrices that were blocked with 1-SA, and compared its performance with that of the cuSPARSE spmm routine.
 In a fashion similar to the experiments of Section~\ref{sec: exp:reorder}, and as detailed in Section~\ref{sec: synthetic-generation}, we generated a landscape of blocked matrices varying their $\theta$, $\rho$ and $\Delta$ values. We scrambled the order of their rows, and then reordered them using 1-SA. 
 After that, we took a reordering with $\Delta_H' \approx \Delta$ as a representative for that landscape point, and we multiplied it with a dense matrix using both VBR-cuBLAS and cuSPARSE. The results, summarized in Figure~\ref{fig:performance_heatmap}, confirm that there is a wide class of matrices for which reordering and blocking consistently speed up the multiplication compared to the  sparse-based baseline. We note that we are able to obtain a speed-up for matrices with overall density as small as $\frac{1}{10.000}$. We observe that the performance improvement grows with the size of the dense matrix being multiplied. This behaviour is explained by noting that, for wider matrices, the computational intensity of the problem increases. This makes the reduced indexing and better data locality of our multiplication routine more effective in reducing the computation time.
 \begin{figure*}[htb!]
    \centering
    \begin{subfigure}[t]{0.32\textwidth} 
        \centering
        \includegraphics[width = \linewidth]{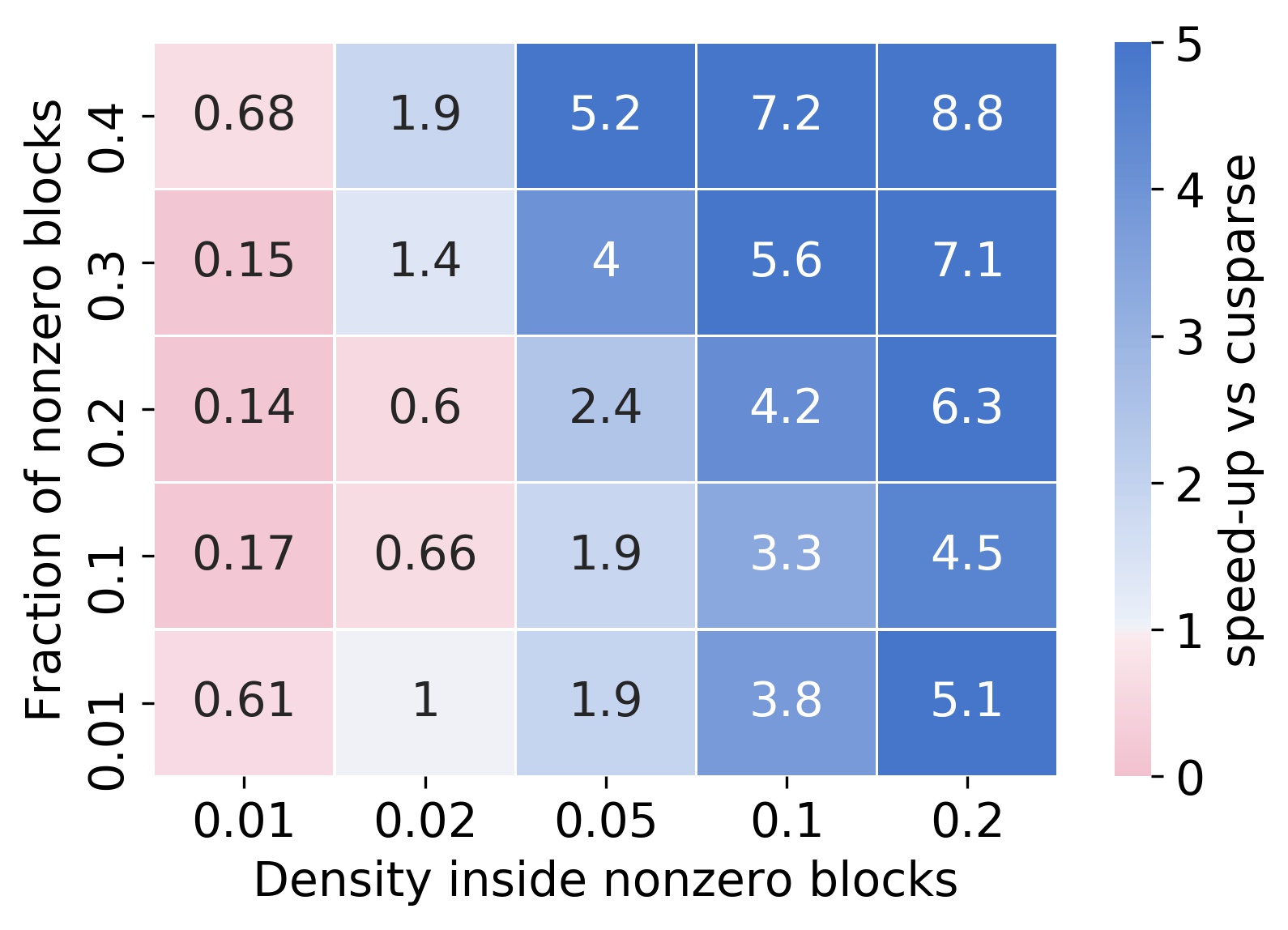}
        \caption{}
    \end{subfigure}
    ~ 
    \begin{subfigure}[t]{0.32\textwidth}     
        \centering
        \includegraphics[width = \linewidth]{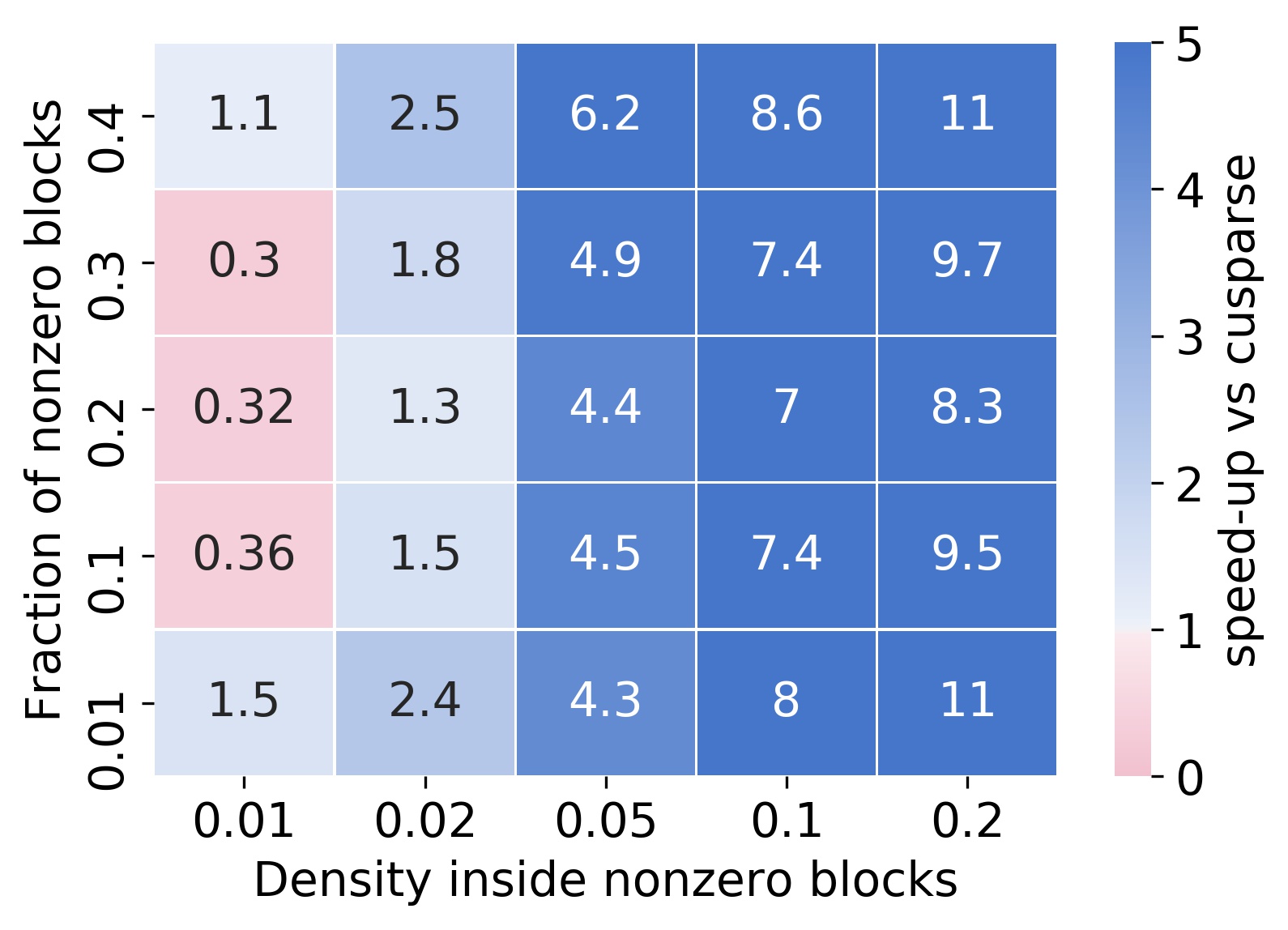}
        \caption{}
    \end{subfigure}
    ~ 
    \begin{subfigure}[t]{0.32\textwidth}     
        \centering
        \includegraphics[width = \linewidth]{ 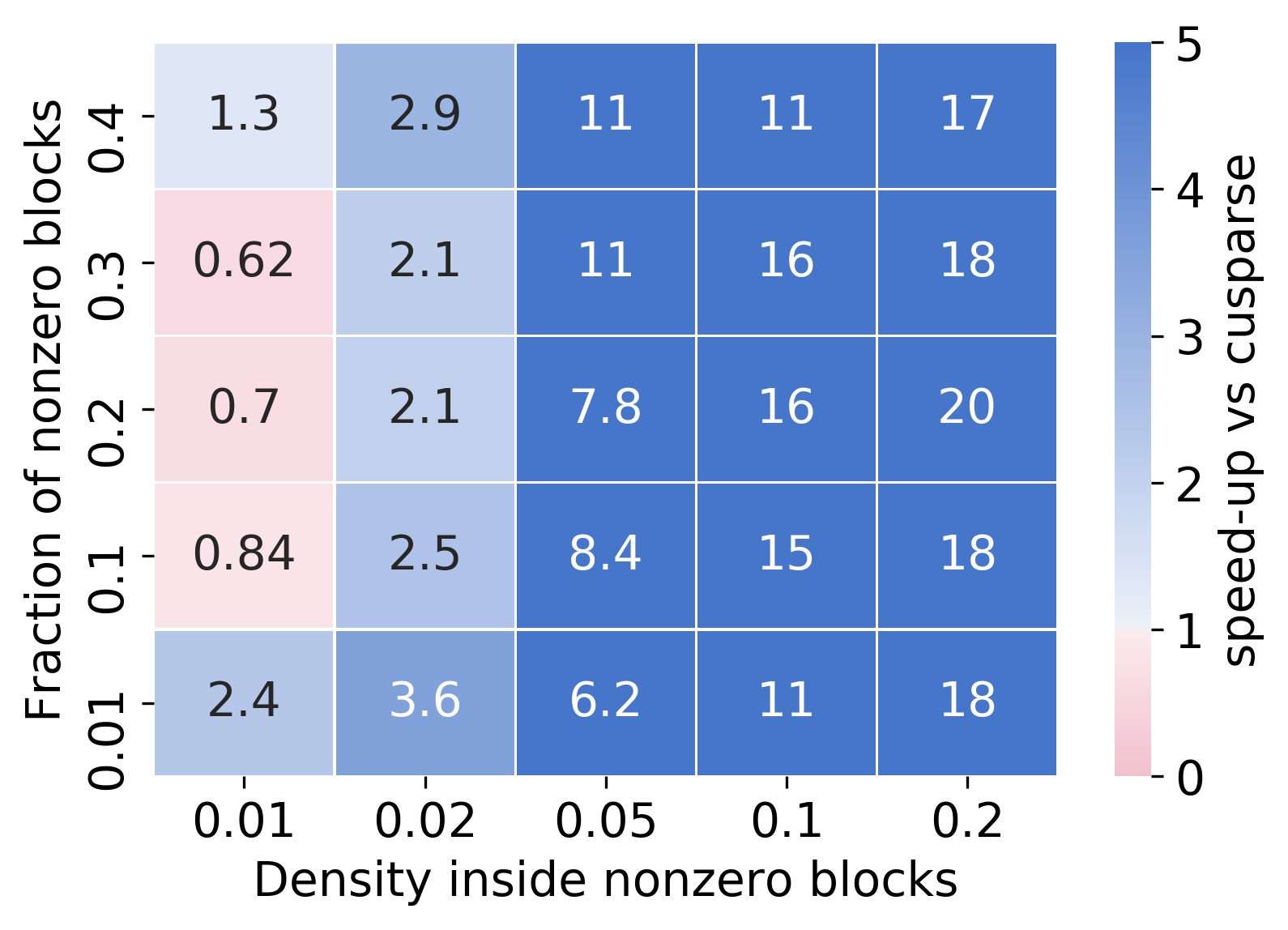}
        \caption{}
    \end{subfigure}

    \caption{Performance improvement of VBR-cuBLAS vs cuSparse for synthetic matrices blocked with 1-SA. The synthetic block matrices are generated as detailed in Sec \ref{sec: synthetic-generation} for each pair of $\rho$ and $\theta$ in the heat-map. The true density of each matrix can be calculated as $\rho \times \theta$. Each $8192 \times 8192$ matrix has been scrambled, blocked with 1-SA, and then multiplied using VBR-cuBLAS and cusparse. Values in the heat-map represent the average ratio between the running time of the two routines over 10 runs. Images (a),(b),(c) show the results when setting the dense matrix width $N$ at $2048$, $4096$ and $8192$ respectively.}
    \label{fig:performance_heatmap}
    
\end{figure*}

\subsubsection{Blocking and multiplication of RMATs}
To better assess the ability of our methodology of speeding up multiplication, we generated RMATs as detailed in Section~\ref{dataset-description}. 
After blocking them with 1-SA, we compared the performance of VBR-cuBLAS and cuSPARSE when multiplying with a dense matrix. In Figure~\ref{fig:barplotRMAT} it can be seen that, for all considered column partition sizes $\Delta_W$, our methodology is effective in reducing the multiplication time. The performance gain grows with the matrix density, and we observe a 4x speed-up for the most dense RMAT. We note that, while using bigger $\Delta_W$ improves performance, the related returns quickly diminish. This is to be expected, since bigger blocks entail higher fill-in.
\begin{figure}
    \centering
    \includegraphics[width = 0.8 \linewidth]{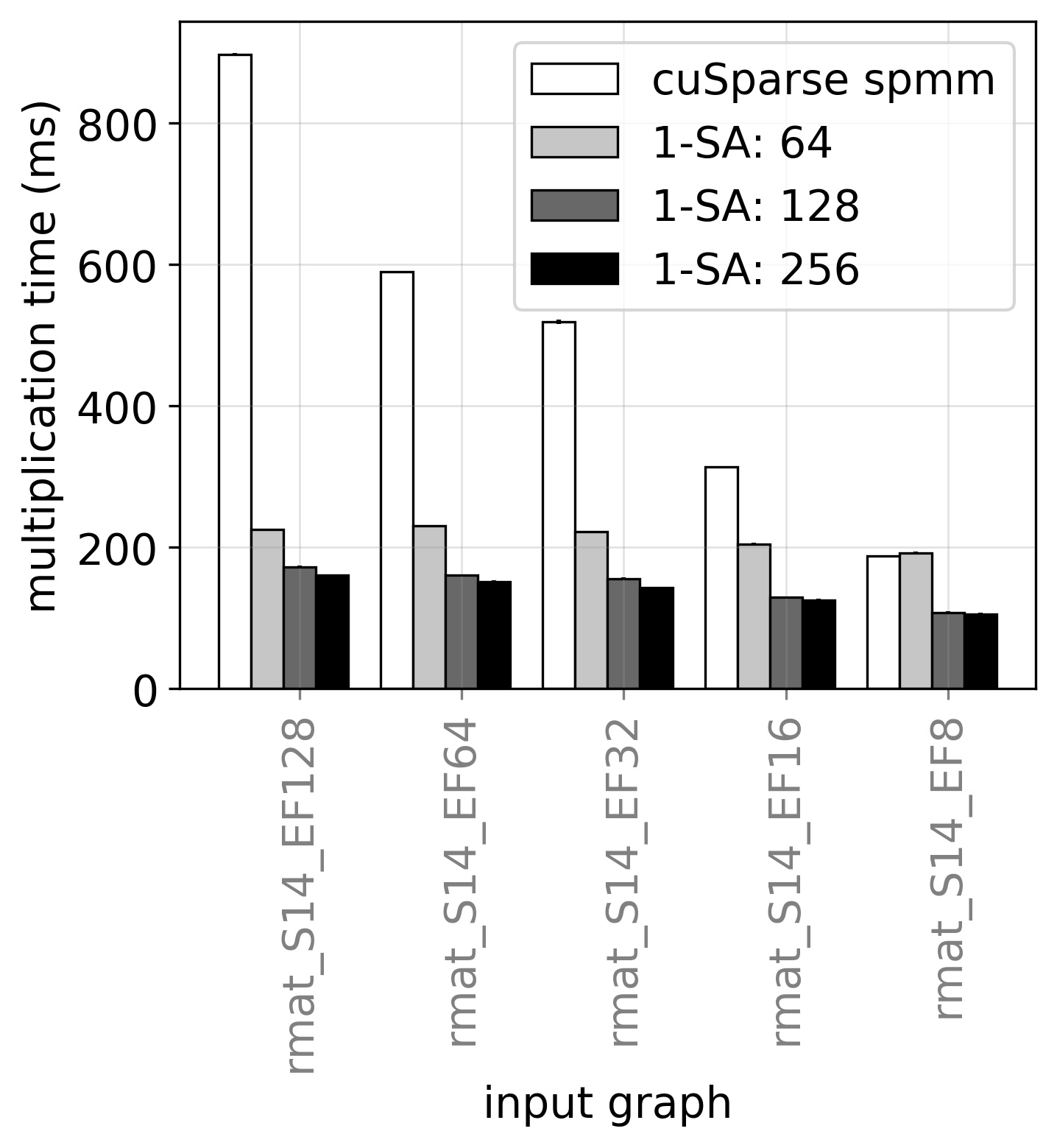}

    \caption{Performance comparison of cuSparse vs 1-SA on synthetic RMATs. RMATs with $2^{14}$ nodes and $8,16,32,64$ or $128$ average edges per node have been generated, scrambled, and blocked with 1-SA, varying the column partition parameter $\Delta_W = 64, 128,256$. Then, the blocked graphs have been multiplied with a dense matrix ($N = 2^{12})$ using VBR-cuBLAS. Each bar shows the average over 10 runs. We omit similar results for other values of $N$.}
    \label{fig:barplotRMAT}
    
\end{figure}

\subsubsection{Blocking and multiplication of real-world matrices}
Finally, we considered the multiplication of real-world sparse matrices from the Network Repository~\cite{netref} database. Details of the employed graphs are provided in Section~\ref{dataset-description}.
Matriced have been blocked with 1-SA, then multiplied with a dense matrix ($N = 4096$) using both VBR-cuBLAS and cuSPARSE. As for the RMATs of the previous section, Figure~\ref{fig:barplot_real} shows that, for all considered column partition sizes $\Delta_W$, our methodology is effective in reducing the multiplication time.

\begin{figure*}[htb!]
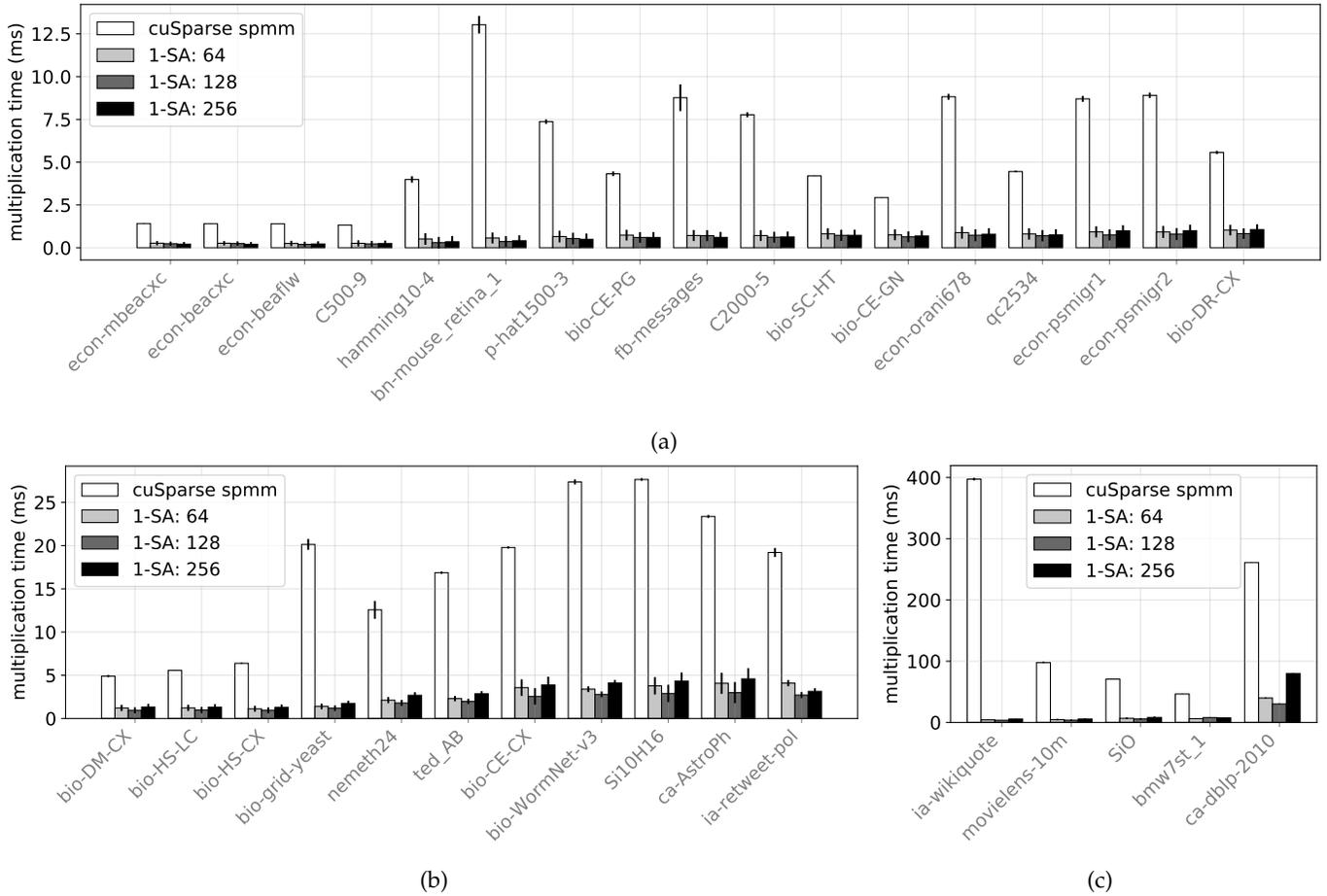

    \centering
    \begin{subfigure}[t]{\textwidth}
        \centering
        \includegraphics[width = \linewidth]{images/small_resaad_blocks_me0_ep0.1_B_4096_hi1.jpg}
        \caption{}
        \label{small-real}
    \end{subfigure}%
    \hfill
    \begin{subfigure}[t]{0.65\textwidth}
        \centering
        \includegraphics[width = \linewidth]{images/medium_resaad_blocks_me0_ep0.1_B_4096_hi1.jpg}
        \caption{}
        \label{medium-real}
    \end{subfigure}
    \begin{subfigure}[t]{0.34\textwidth}
    \centering
    \includegraphics[width = \linewidth]{images/big_resaad_blocks_me0_ep0.1_B_4096_hi1.jpg}
    \caption{}
    \label{big-real}
\end{subfigure}
    
    \caption{Performance comparison of cuSparse vs 1-SA on real-world sparse matrices from the Network Repository\cite{netref} collection. (a),(b) and (c) show three groups of graphs of growing size. See Sec \ref{dataset-description} for the graphs description. The matrices have been scrambled and blocked with 1-SA, varying the column partition parameter $\Delta_W = 64, 128,256$. Then, the blocked matrices have been multiplied with a dense matrix ($N = 4096$) using both VBR-cuBLAS and cusparse. Each bar shows the average over 10 runs.}
    \label{fig:barplot_real}
    
\end{figure*}

\section{Related Works}
\label{sec:related}
The idea of reordering sparse matrices to find their substructures or reorganize their data has been extensively explored, and several algorithms exist to do so with different purposes and efficacy~\cite{CMReorder}\cite{AshMinimumDegree}\cite{Saad}\cite{ApproxMinimumDegree}\cite{DistanceReorder}\cite{PichelReorder}.
One of the first, and most common use of reordering algorithms is as preconditioners for iterative solvers~\cite{benzi2002preconditioning}.
Typical reordering algorithms in this class are the Reversed Cuthill-Mckee (RCM)~\cite{cuthill1972} and the Approximate Minimum Degree (AMD)~\cite{ApproxMinimumDegree}. One consequence of this origin is that most reordering methods work on structurally symmetric matrices, and produce symmetric reorderings since the latter will not change the diagonal entries and the eigenvalues of a matrix. Apart from preventing the use of these algorithms on asymmetric or rectangular sparse matrices, such as those found in pruned neural networks, these assumptions also exclude asymmetric reorderings, which could provide better blocking. 
While many algorithms were developed with the purpose of blocking, such as the popular PABLO~\cite{PABLO} and its variants, they would only consider diagonal blocking. Ashcraft's~\cite{AshMinimumDegree} and Saad's~\cite{Saad} compression methods, the starting point of our research, are an exception in this regard in that they also promote the creation of off-diagonal blocks. 
Another important class of reordering and blocking algorithms comes from graph partitioning problems. In this regard, popular algorithms are nested dissection algorithms~\cite{nestedDissection}, multilevel algorithms~\cite{multilevel}, greedy methods~\cite{recallReorder}, spectral partitioning, and metaheuristics such as simulated annealing or genetic algorithms. Partitioning algorithms can also be extended to treat asymmetric and rectangular matrices, by considering instead bipartite graphs or hypergraphs~\cite{catalyurek1999hypergraph}. While these algorithms are more focused on pure blocking than the ones coming from preconditioning, they focus on minimizing the cut produced by the partition, that is, the total weight of edges connecting partitions. This means that they will produce too an approximately optimal diagonal blocking, but not necessarily a good blocking for block-based multiplication in general.

To the best of our knowledge, the work closer to our own is described by Zachariadis et al.~\cite{accTensor2020}.
They considered the problem of accelerating sparse-sparse matrix multiplication on Nvidia Tensor Cores. In their work, they divided (without reordering) the matrices in tiles, and multiply only nonzero tiles by using auxiliary data structures on the fly, with the consequence of increasing load-and-store operations. Under such a light, their techniques can also be integrated with 1-SA to further increase the density inside of hyper-sparse blocks. In general, their approach could be integrated with ours to treat efficiently even sparser matrices than the one we considered in this paper.

Other flavours of blocking and reordering have also been investigated to accelerate parallel spMM. Pichel et al.~\cite{PichelComparison} reordered several matrices with different reordering algorithms and tested their GPU spMV performance in the CSR, HYB, ELLPACK and BELLPACK formats. They found that most of the time, a reordering would exist that considerably improve performance. They provided some insights on which reordering works best for which format, stating for example that BELLPACK benefits from "distance-based" reorderings and is harmed by bandwidth-reduction algorithms. Yet, in that work, they neither provided a way to find a good reordering for a given matrix nor explored the trade-off between block size and block sparsity.
In previous work, Pichel et al. developed a framework to reorder a sparse matrix in order to cluster entries in the same row~\cite{PichelReorder}; their approach uses an approximated algorithm for the Travelling Salesman Problem (TSP) to maximize locality, i.e., the number of nonzero entries appearing in consecutive locations. In this regard, their approach is similar to Saad's.

Pinar et al.~\cite{Pinar1} also considered a problem close to our own. They defined a storage format based on 1D horizontal blocks (BCRS) to perform spMV on a CPU. Then, they set themselves to reorder the columns of a matrix to maximize such blocking. Having reduced this problem to a TSP, they used TSP heuristics to solve it. Interestingly, they employed very small blocks (1x2 or 1x3) and reported no benefit from using larger blocks. 

Hong et al. observed~\cite{spMMHong} the benefit of identifying clustered entries for spMM on the GPU, devising a hybrid data structure (CSR + row-blocked CSR) that treats separately blocked and non-blocked entries. In a later work~\cite{spMMHongTiling}, they transposed this approach directly to the multiplication phase, avoiding the creation of a non-standard storage format and relying instead on the classic CSR. To improve blocking, they used an ordering strategy (but not a renumbering) that, within a tile, clusters together columns that exceed a certain density. However, their tiles are only used to organize the multiplication, and they still use the same kernel to multiply both sparse and dense tiles.

\section{Conclusion}
\label{sec:conclusion}
We presented a 1-dimensional blocking algorithm able to decompose a sparse matrix in dense blocks, with tunable block size and guaranteed block density. The resulting dense blocks can be easily processed by tensor accelerators and used to speed up the sparse-by-dense matrix multiplication on these architectures. We focused on the Nvidia CUDA hardware/software stack and showed evidence that our approach is several times faster than the sparse-specific cuSPARSE routine.

We evaluated our methodology in the popular CUDA framework, but we note that it does not require specific libraries or architectures to be implemented.
We suggest that this methodology can be employed in all situations where a highly efficient, block-based matrix multiplication routine or architecture exists to bypass the need for low-level sparse-specific solutions.

As a final remark, we did not study in this work the cost of blocking, operating under the assumption that it is discounted over many multiplications. However, we report that our current {\em naive} serial implementation of 1-SA takes roughly the same time to reorder the matrix as cuBLAS and cuSPARSE take to multiply it. In future works, we expect that optimizing and possibly parallelizing 1-SA will make run-time blocking convenient even for a single multiplication instance.

\ifCLASSOPTIONcaptionsoff
  \newpage
\fi



%\bibliographystyle{IEEEtran}

\bibliographystyle{IEEEtran}

\bibliography{mainmain}


\end{document}